\documentclass{article}

\PassOptionsToPackage{numbers, sort&compress}{natbib}

\usepackage[preprint]{neurips_2022}
\usepackage{amsmath, amssymb, amsthm}

\usepackage[utf8]{inputenc} 
\usepackage[T1]{fontenc}    
\usepackage[hidelinks,colorlinks,allcolors=magenta]{hyperref}       
\usepackage{url}            
\usepackage{booktabs}       
\usepackage{amsfonts}       
\usepackage{nicefrac}       
\usepackage{microtype}      
\usepackage{xcolor}         

\usepackage{booktabs} 
\usepackage[ruled]{algorithm2e} 

\usepackage{algorithmic}

\SetAlFnt{\small}
\SetAlCapFnt{\small}
\SetAlCapNameFnt{\small}
\SetAlCapHSkip{0pt}
\IncMargin{-\parindent}

\mathcode`l="8000
\begingroup
\makeatletter
\lccode`\~=`\l
\DeclareMathSymbol{\lsb@l}{\mathalpha}{letters}{`l}
\lowercase{\gdef~{\ifnum\the\mathgroup=\m@ne \ell \else \lsb@l \fi}}%
\endgroup

\setcitestyle{number}

\usepackage{nicefrac} 
\usepackage{nicematrix}
\usepackage{amsmath,amsthm,amsfonts,dsfont,bbm}
\usepackage{comment}
\usepackage{graphicx}
\usepackage{tcolorbox}
\usepackage{enumitem}
\usepackage{wrapfig}
\usepackage{xcolor}
\usepackage{thm-restate}
\usepackage[capitalize]{cleveref}
\usepackage[colorinlistoftodos]{todonotes}
\usepackage{url}            
\usepackage{mathtools}
\usepackage[multiple]{footmisc}
\usepackage{xspace}

\def\algoname{Set-Aside Greedy\xspace}

\renewcommand{\tilde}{\widetilde}

\renewcommand{\le}{\leqslant}
\renewcommand{\leq}{\leqslant}
\renewcommand{\ge}{\geqslant}
\renewcommand{\geq}{\geqslant}
\newcommand{\one}{\mathds{1}}
\newcommand{\bbN}{\mathbb{N}}

\newcommand{\bx}{\mathbf{x}}
\newcommand{\wv}{\mathbf{w}}
\newcommand{\xv}{\mathbf{x}}
\newcommand{\zv}{\mathbf{z}}
\newcommand{\yv}{\mathbf{y}}
\newcommand{\bv}{\mathbf{v}}

\newcommand{\NSW}{\mathsf{NSW}}

\usepackage{multicol}
\usepackage{diagbox}
\newcommand\abs[1]{\left\lvert{#1}\right\rvert}
\newcommand{\set}[1]{\left\{#1\right\}}
\newcommand{\R}{\mathbb{R}}

\newcommand{\thetav}{\boldsymbol\theta}

\newtheorem{theorem}{Theorem}
\newtheorem{proposition}{Proposition}
\newtheorem{corollary}{Corollary}
\newtheorem{lemma}{Lemma}

\theoremstyle{definition}

\newtheorem{definition}{Definition}

\title{Proportionally Fair Online Allocation\\of Public Goods with Predictions}

\author{%
Siddhartha Banerjee\\Cornell University\\\texttt{sbanerjee@cornell.edu}
\And
Vasilis Gkatzelis\\Drexel University\\\texttt{gkatz@drexel.edu}
\And
Safwan Hossain\\University of Toronto\\\texttt{safwan.hossain@mail.utoronto.ca}
\AND
Billy Jin\\Cornell University\\\texttt{bzj3@cornell.edu}
\And
Evi Micha\\University of Toronto\\\texttt{emicha@cs.toronto.edu}
\And
Nisarg Shah\\University of Toronto\\\texttt{nisarg@cs.toronto.edu}
}

\begin{document}

\maketitle

\begin{abstract}
    We design online algorithms for the fair allocation of public goods to a set of $N$ agents over a sequence of $T$ rounds and focus on improving their performance using predictions. In the basic model, a public good arrives in each round, the algorithm learns every agent’s value for the good, and must irrevocably decide the amount of investment in the good without exceeding a total budget of $B$ across all rounds. The algorithm can utilize (potentially inaccurate) predictions of each agent’s total value for all the goods to arrive. We measure the performance of the algorithm using a \emph{proportional fairness} objective, which informally demands that every group of agents be rewarded in proportion to its size and the cohesiveness of its preferences. 

    In the special case of binary agent preferences and a unit budget, we show that $O(\log N)$ proportional fairness can be achieved without using any predictions, and that this is optimal even if perfectly accurate predictions were available. However, for general preferences and budget no algorithm can achieve better than $\Theta(T/B)$ proportional fairness without predictions. We show that algorithms with (reasonably accurate) predictions can do much better, achieving $\Theta(\log (T/B))$ proportional fairness. We also extend this result to a general model in which a batch of $L$ public goods arrive in each round and achieve $O(\log (\min(N,L) \cdot T/B))$ proportional fairness. Our exact bounds are parametrized as a function of the error in the predictions and the performance degrades gracefully with increasing errors.

\end{abstract}

\section{Introduction}

In classic online algorithms, the input is presented in stages and the algorithm needs to make irrevocable decisions at each stage without knowing the input from future stages. Its performance, called the competitive ratio, is measured by comparing the worst-case ratio of the achieved solution quality to the optimal solution quality in hindsight~\cite{borodin2005online}. The uncertainty regarding the future often forces such algorithms to make overly cautious decisions, resulting in pessimistic competitive ratios.

An emerging line of research asks whether one can improve the performance of online algorithms using predictions regarding the future~\cite{mitzenmacher_vassilvitskii_2021}.
Ideally, we would like the algorithm to perform well when the predictions are good, yet maintain reasonable performance even when the predictions are bad.
 More generally, one can hope to express the competitive ratio of the online algorithm directly in terms of the error in the predictions. This powerful paradigm has already received significant attention, for problems such as caching~\citep{lykouris2018competitive,rohatgi2020near,jiang2020online}, the secretary problem~\citep{dutting2021secretaries,antoniadis2020online,antoniadis2020secretary}, scheduling~\citep{lattanzi2020online}, the ski rental problem~\citep{purohit2018improving,wang2020online,banerjee2020improving}, set cover~\citep{bamas2020primal}, and other problems~\citep{almanza2021online,antoniadis2021learning,dinitz2021faster,yu2020power,sun2021pareto,im2021online,bamas2020learning,anand2021regression,li2019online}. 

However, all of this work is limited to \emph{single-agent} decision-making problems. Recently, \citet{BGGJ22} applied this paradigm to design online algorithms for a \emph{multi-agent} resource allocation problem, in which a set of \emph{private} goods (which can only be allocated to and enjoyed by a single agent) need to be divided amongst a group of agents in a \emph{fair} manner. This problem has direct applications to settings like inheritance division and divorce settlement, and has garnered interest from various research communities, dating back to \citet{steihaus1948problem}. Using the \emph{Nash welfare} from bargaining theory~\citep{Nash50} as their notion of fairness, \citet{BGGJ22} show that predictions about agents' total value for all goods can be utilized to achieve significantly improved approximations.  

The solutions proposed by \citet{BGGJ22}, however, do not capture
resource allocation settings involving \emph{public} goods, i.e., goods whose benefit can be enjoyed by multiple agents (e.g., a highway or a park). In many important problems, like participatory budgeting, committee selection, or shared memory allocation, some scarce resources need to be dedicated to make each public good available, and an algorithm needs to decide which goods to invest in, aiming to make the agents happy. Fairness in these settings is often captured by notions like the Nash welfare and the \emph{core}~\cite{foley1970lindahl}. Yet, despite the significance of its applications, only a few papers have successfully studied the fair allocation of public goods~\cite{fain2016core,KFMB17,FGPS19,CFS17,FMS18,PPS20,MSWW} (relative to the extensive literature on private goods; see~\citep{comsocfairdivision}), even fewer have provided \emph{online} algorithms for this problem~\cite{FZC17}, and none utilize machine-learned predictions. 



We address this gap by designing online algorithms for fair allocation to public goods using predictions about how agents value the goods. The key research questions we address are:
\begin{quote}
    \emph{How can we allocate public goods in an online yet fair manner? To what extent can predictions about agent preferences help improve the fairness guarantees?}
\end{quote}

\subsection{Our Results}
We study the design of online algorithms for the fair allocation to public goods arriving over a sequence of $T$ rounds based on the preferences of $N$ agents. In the basic model a new public good in each round $t$ (which we refer to as good $t$),  and the algorithm learns the value $v_{i,t}$ of each agent $i$ for this good. Using this information, the algorithm needs to make an irrevocable decision regarding an amount $x_t \in [0,1]$ to invest in this good. This value $x_t$ can be interpreted as the amount of a scarce resource (e.g., money) devoted toward making this good available, or the probability of implementing the good. Each agent $i$ then receives value $v_{i,t} \cdot x_t$ from this investment, and the total value derived by an agent is the sum of the values gained across rounds. 

Clearly, the algorithm would like to increase the $x_t$'s as much as possible, but the algorithm is limited by a total budget constraint: $\sum_t x_t \le B$, where $B$ is given. Informally, the budget constraint forces the algorithm to invest more in goods that are highly valued by many agents. However, this is challenging since the agents' values for future goods are unknown. To deal with this uncertainty, the algorithm has access to \emph{predictions} regarding the \emph{total value} of each agent $i$: this prediction $\tilde{V}_i$ provides an estimate on the total value $V_i = \sum_t v_{i,t}$ of agent $i$ for all goods to arrive. 

We focus on a quantitative fairness objective, called \emph{proportional fairness} (\Cref{def:pf}), which is stronger than previously considered objectives of the Nash welfare and the core (see \Cref{subsec:prop_fairness}). Lower objective values indicate better fairness, with $1$ indicating perfect proportional fairness. 

In \Cref{sec:binary}, as a warm-up, we consider the setting where agent values are binary (i.e., $v_{i,t}\in \set{0,1}$ for all agents $i$ and goods $t$) and the budget is $B=1$. Binary values correspond to \emph{approval voting} -- where agents either like a good or not; the unit budget forces $\sum_t x_t \le 1$, meaning that $x_t$ can also be viewed as the fraction of an available resource invested in good $t$. For this special case, we show that it is already possible to achieve $O(\log N)$ proportional fairness without using any predictions and this is optimal even if the algorithm had access to \emph{perfect predictions} ($\tilde{V}_i = V_i$ for each agent $i$). 

In \Cref{sec:single_public}, we consider general values and budget, and show that this is no longer true: without access to predictions, no algorithm can achieve $o(T/B)$ proportional fairness even when with $N=1$ agent. Our main positive result shows that, by using the predictions, we can achieve an exponential improvement to $O(\log (T/B))$ proportional fairness for $N$ agents, as long as the predictions are reasonably accurate (constant multiplicative error). We also show this to be optimal given predictions. 

Finally, in \Cref{sec:batched} we extend our model even further by allowing a batch of $L$ public goods to arrive in each round $t$ and achieve $O(\log (\min(N,L) \cdot T/B))$ proportional fairness. In fact, we show that this model strictly generalizes not only our initial public-goods model, but also the private-goods model of \citet{BGGJ22}, and our positive result in this very broad model (almost) implies theirs. 

\subsection{Related Work}

\textbf{Allocation of private goods.} The majority of the prior work on online fair division has focused on \emph{private} goods, for which achieving even basic notions of fairness comes at the cost of extreme inefficiency in the absence of any predictions regarding agents' values for future goods~\citep{BKPP18,ZP20}. \citet{BGGJ22} show that total value predictions can be leveraged to achieve improved fairness guarantees with respect to the Nash social welfare objective. Note that assuming access to predictions of agents' total values for all goods is related to work which assumes that the values of each agent are normalized to add up to $1$~\cite{GPT21,BKM21} or that they are drawn randomly from a normalized distribution~\cite{BMS19}. Our work reinforces the findings of \citet{BGGJ22} that predictions help significantly improve fairness guarantees, but in our more general model with public goods. 

\textbf{Allocation to public goods.} Much of the literature on public goods focuses on the offline setting -- agents have approval preferences, and we have a fixed budget (i.e., the offline version of the setting in \Cref{sec:binary}). This offline version has been studied under various names, such as probabilistic voting~\cite{BMS05a}, fair mixing~\cite{ABM17a}, fair sharing~\cite{Dudd15a}, portioning~\cite{BBPS21a}, and randomized single-winner voting~\cite{ebadian2022optimized}. 
\citet{ebadian2022optimized} show that with access to only \emph{ranked} preferences, $\Theta(\log T)$ is the best possible proportional fairness in the offline setting. Interestingly, we show that incomplete information resulting from online arrivals leads to the same $\Theta(\log T)$ proportional fairness.

Multi-winner voting extends this by selecting a subset of $k$ candidates, which is the offline version of our model in \Cref{sec:single_public} with $B=k$. 
Multi-winner voting can be extended to fair public decision-making~\cite{CFS17} and participatory budgeting~\cite{aziz2021participatory,fain2016core}, which are further generalized by the public good model of \citet{FMS18}. These generalizations allow more complex feasibility restrictions on the allocations than our budget constraint, but work in the offline setting. To the best of our knowledge, the only work to consider \emph{online} allocation to public goods is that of \citet{FZC17}, who consider optimizing the Nash welfare in a model similar to ours. However, they do not provide any approximation guarantees; instead, they study natural online algorithms from an axiomatic viewpoint. 

\textbf{Primal-dual analysis.} Finally, we remark that our main positive result is derived using a primal-dual-style analysis (see the survey of~\citet{buchbinder2009design} for an excellent overview). Almost all of this work deals with additive objective functions. Two notable exceptions to this are the work of~\cite{DJ2012}, who show how to extend these approaches to non-linear functions of additive rewards, as well as~\citet{azar2016allocate}, who consider a variant of the proportional allocation objective, but require additional boundedness conditions on the valuations. \citet{bamas2020primal} show how to adapt primal-dual analyses to incorporate predictions for a number of single-agent problems.  


\section{Model}
\label{sec:model}

We study an online resource allocation problem with $N$ agents and $T$ rounds. Our algorithms observe the number of agents $N$, but they may not know the number of rounds $T$, in advance. We study both \emph{horizon-aware} algorithms which know $T$ and \emph{horizon-independent} algorithms which do not. For simplicity we use $[k]$ to denote the set $\set{1,\ldots,k}$ for a given $k \in \mathbb{N}$.

\emph{Online arrivals.} In the basic version of the model, in each round $t \in [T]$, a public good, which we refer to as good $t$, ``arrives''. Upon its arrival, we learn the value $v_{i,t} \ge 0$ of every agent $i \in [N]$ for it. In \Cref{sec:batched}, we extend the model to allow a batch of $L$ public goods arriving in each round. 

\emph{Online allocations.} When good $t$ arrives, the online algorithm must irrevocably decide the allocation $x_t \in [0,1]$ to good $t$, before the next round.\footnote{In this work we focus on \emph{deterministic} algorithms, but this is w.l.o.g.\ since we consider fractional allocations.} We use $\xv = (x_t)_{t \in [T]}$ to denote the final allocation computed by the online algorithm. In the absence of any further constraints, the decision would be simple: allocate as much as possible to every good by setting $x_t = 1$ for each $t \in [T]$. However, the extent to which the algorithm can allocate these goods is limited by an overall budget constraint: $\sum_{t=1}^T x_t \le B$, where $B \ge 0$ is a fixed budget known to the online algorithm in advance.

\emph{Linear agent utilities.} Choosing to allocate $x_t$ to good $t$ simultaneously yields utility $v_{i,t} \cdot x_t$ to every agent $i\in [N]$. Moreover, we assume that agent utilities are additive across goods, i.e., the final utility of agent $i$ is given by $u_i(\xv) = \sum_{t=1}^T v_{i,t} \cdot x_t$. This class of linear agent utilities is widely studied and it admits several natural interpretations, depending on the application of interest. In applications like budget division, each public good $t$ is a project (e.g., an infrastructure project), and $x_t$ is the amount of a resource (e.g., time or money) invested in the project. In applications such as participatory budgeting or voting, each public good $t$ is an alternative or a candidate, and $x_t$ is the (marginal) probability of it being selected (indeed, one can compute a lottery over subsets of goods of size at most $B$ under which the marginal probability of selecting each good $t$ is precisely $x_t$). 

When working with fractions $x/y$ with $x,y\geq 0$, we adopt the convention that $x/y=1$ when both $x=y=0$, while $x/y=+\infty$ if $y=0$ and $x>0$. We use $H_k = 1+\frac{1}{2}+\frac{1}{3}+\ldots+\frac{1}{k}$ to denote the $k$-th harmonic number. 

\subsection{Proportional Fairness}
\label{subsec:prop_fairness}
We want the allocation $\xv$ computed by our online algorithm to be \emph{fair}. In this work, we use the notion of proportional fairness, which is a quantitative fairness notion that was first proposed in the context of rate control in communication networks~\cite{kelly1998rate}. 

\begin{definition}[Proportional Fairness]\label{def:pf}
For $\alpha \ge 1$, allocation $\xv$ is called $\alpha$-proportionally fair if for every other allocation $\wv$ we have $\frac{1}{N} \sum_{i=1}^N \frac{u_i(\wv)}{u_i(\xv)} \le \alpha$. If $\xv$ is $1$-proportionally fair, we simply refer to it as proportionally fair\footnote{The $1$-proportional fair criterion is more commonly (but equivalently) written as $\frac{1}{N} \sum_{i=1}^N \frac{u_i(\wv)-u_i(\xv)}{u_i(\xv)} \le 0$.}.
We say that an online algorithm is $\alpha$-proportionally fair if it always produces an $\alpha$-proportionally fair allocation. 
\end{definition}

It is known that in the offline setting, where all agent values are known up front, a $1$-proportionally fair allocation $\xv$ always exists, and this is the lowest possible value of proportional fairness~\cite{FMS18}. It is also known that proportional fairness is a strong guarantee that implies several other guarantees sought in the literature. Below, we show two examples: the core and Nash social welfare. 

\emph{Proportional fairness implies the core.} For $\alpha \ge 1$, allocation $\xv$ is said to be in the $\alpha$-core if there is no subset of agents $S$ and allocation $\wv$ such that $\frac{|S|}{N} \cdot u_i(\wv) \ge \alpha \cdot u_i(\xv)$ for all $i \in S$ and at least one of these inequalities is strict. We say that an online algorithm is $\alpha$-core if it always produces an allocation in the $\alpha$-core. The following is a well-known relation between proportional fairness and the core. 
\begin{proposition}
\label{prop:core}
For $\alpha \ge 1$, every $\alpha$-proportionally fair allocation is in the $\alpha$-core.
\end{proposition}
\begin{proof}
If an $\alpha$-proportionally fair allocation $\xv$ is not in the $\alpha$-core, then by definition, there exists a subset of agents $S$ and an allocation $\wv$ such that $\frac{|S|}{N} u_i(\wv) \ge \alpha \cdot u_i(\xv)$ for all $i \in S$ and at least one of these inequalities is strict. Then, $\frac{u_i(\wv)}{u_i(\xv)} \ge \alpha \cdot \frac{N}{|S|}$ for all $i \in S$ and at least one of these inequalities is strict, which implies 
\[
\sum_{i \in S} \frac{u_i(\wv)}{u_i(\xv)} > \alpha \cdot N \Rightarrow \frac{1}{N} \sum_{i=1}^N \frac{u_i(\wv)}{u_i(\xv)} \ge \frac{1}{N} \sum_{i \in S} \frac{u_i(\wv)}{u_i(\xv)} > \alpha,
\]
contradicting the fact that $\xv$ is $\alpha$-proportionally fair. 
\end{proof}

\emph{Proportional fairness implies optimal Nash welfare.} A common objective function studied in multi-agent systems is the Nash social welfare, which aggregates individual agent utilities into a collective measure by taking the geometric mean. That is, the Nash social welfare of allocation $\xv$ is given by $\NSW(\xv) = \left(\prod_{i=1}^N u_i(\xv)\right)^{1/N}$. For $\alpha \ge 1$, we say that allocation $\xv$ achieves an $\alpha$-approximation of the Nash welfare if $\frac{\NSW(\wv)}{\NSW(\xv)} \le \alpha$ for all allocations $\wv$. We say that an online algorithm achieves an $\alpha$-approximation of the Nash welfare if it always produces an allocation that achieves an $\alpha$-approximation of the Nash welfare. It is also well-known that $\alpha$-proportional fairness implies an $\alpha$-approximation of the Nash welfare (in particular, a proportionally fair allocation has the maximum possible Nash welfare). 

\begin{proposition}
\label{prop:nsw}
For $\alpha \ge 1$, if allocation $\xv$ is $\alpha$-proportionally fair, then $\xv$ achieves an $\alpha$-approximation of the Nash welfare. 
\end{proposition}
\begin{proof}
Take any allocation $\wv$. The result follows by observing that
\[
\frac{\NSW(\wv)}{\NSW(\xv)} = \left( \prod_{i=1}^N \frac{u_i(\wv)}{u_i(\xv)} \right)^{1/N} \le \frac{1}{N} \sum_{i=1}^N \frac{u_i(\wv)}{u_i(\xv)} \le \alpha,
\]
where the second transition is the AM-GM inequality.  
\end{proof}

We remark that the upper bounds derived in this work hold for the stronger notion of proportional fairness, while the lower bounds hold even for the weaker notion of Nash welfare approximation.

\subsection{Set-Aside Greedy Algorithms}\label{subsec:setaside}

In order to compute (approximately) proportionally fair allocations, we consider a family of online algorithms, called \emph{Set-Aside Greedy Algorithms}.  Recent work~\citep{BGGJ22,BKM21} has demonstrated how such algorithms can be used to get strong performance guarantees for online allocation of \emph{private goods}; we show that with non-trivial modifications, they can also achieve compelling fairness guarantees for allocating public goods. 

At a high level, an algorithm in this family divides the overall budget $B$ into two equal portions. 
\begin{enumerate}
    \item The first half, called the \emph{set-aside budget}, is used to allocate $y_t \in [0,1]$ to each good $t$ in such a manner that $\sum_{t=1}^T y_t \le B/2$ and this portion of the allocation guarantees each agent $i$ a certain minimum utility of $\Delta_i$ (i.e., $\sum_{t=1}^T v_{i,t} \cdot y_t \ge \Delta_i$). For example, if $y_t = B/(2T)$ for each $t \in [T]$, then we can use $\Delta_i = \frac{B}{2T} \cdot \sum_{t=1}^T v_{i,t}$. This ensures that in the proportional fairness definition (\Cref{def:pf}), the ratio $\frac{u_i(\wv)}{u_i(\xv)}$ is not excessively large for any agent $i$. 
    \item The second half, called the \emph{greedy budget}, is used to allocate $z_t \in [0,1-y_t]$ to each good $t$ in such a manner that $\sum_{t=1}^T z_t \le B/2$. This portion of the budget is used in a adaptive greedy-like fashion toward online optimization of the desired objective. 
\end{enumerate}

We refer to $y_t$ and $z_t$ as \emph{semi-allocations} to good $t$, and the final allocation to good $t$ is determined by combining these two semi-allocations, i.e., $x_t = y_t+z_t$. An important quantity in both our algorithm design and its analysis is the promised utility to an agent, defined next. 

\begin{definition}[Promised Utility]
The semi-allocations $y_1,\ldots,y_T$ guarantee that by the end of the algorithm each agent $i$ will receive a utility of $\sum_{t=1}^T v_{i,t} \cdot y_t \geq \Delta_i$ from the set-aside portion of the budget. By round $t$ the algorithm has already set semi-allocations $z_1,\ldots,z_{t-1}$ using the greedy portion of the budget, and needs to now decide $z_t$. At this stage, as a function of $z_t$, the algorithm can guarantee that each agent $i$ will eventually receive utility at least 
$\tilde{u}_{i,t}(z_t) =\Delta_i + {\textstyle\sum_{\tau=1}^t} v_{i,\tau} \cdot z_{\tau}$, even if they do not benefit from any more of the greedy budget. We refer to this as the \emph{promised utility}.
\end{definition}



\section{Warm Up: Binary Utilities and Unit Budget}\label{sec:binary}

Before presenting our main results, we first build some intuition regarding our online setting and the proportional fairness objective by considering the interesting special case where agents have binary utilities for goods (i.e., $v_{i,t} \in \set{0,1}$ for each $i,t$) and the total budget is $B=1$. In this setting, which is motivated by approval voting, we say that agent $i$ ``likes'' good $t$ if $v_{i,t} = 1$, and does not like good $t$ otherwise. Note that with $B=1$, the budget constraint is $\sum_{t=1}^T x_t \le 1$, which means $x_t$ can be interpreted as the \emph{fraction} of an available resource (e.g., time or money) that is dedicated to good $t$. 

On the negative side, we show that no online algorithm can achieve $o(\log N)$-proportional fairness, or even the weaker guarantee of $o(\log N)$-approximation of the Nash welfare. On the positive side we provide a set-aside greedy algorithm that achieves $O(\log N )$ proportional fairness (and
therefore, $O(\log N )$-NSW optimality), thus establishing $\Theta(\log N )$ as the best possible approximation
in this restricted scenario. 

First, in the trivial case with a single agent ($N=1$), we can simply set $x_t = 1$ when the first good $t$ liked by the agent arrives,\footnote{If the agent does not like any good, (exact) proportional fairness is trivially achieved.} which easily yields (exact) proportional fairness. 

It is tempting to extend this idea to the case of $N > 1$ agents. However, we find that even in this restricted scenario with binary utilities and unit budget, no online algorithm achieves $o(\log N)$-proportional fairness, or even the weaker guarantee of $o(\log N)$-approximation of the Nash welfare. In fact, this remains true even if the algorithm is horizon-aware (i.e., knows $T$ in advance) and knows precisely how many goods each agent will like in total. Intuitively, this is because we show that no online algorithm can sufficiently distinguish between instances where many agents like the same goods and those where agents like mostly disjoint goods. 

\begin{theorem}\label{thm:lower-bound-binary}
With binary agent utilities ($v_{i,t} \in \set{0,1}, \forall i,t$) and unit budget ($B=1$), every online algorithm is $\Omega(\log N)$-proportionally fair (in fact, achieves $\Omega(\log N)$-approximation of the Nash welfare), even if the algorithm is horizon-aware and knows in advance the total number of goods each agent will like. 
\end{theorem}

Before we prove the theorem, we need the following technical lemma. Recall that for $k \in \bbN$, $H_k$ is the $k$-th harmonic number. 
\begin{lemma}\label{lem:lower-bounds}
For every $S \in \bbN$, $W \ge 1$, and $\yv=(y_{\ell})_{\ell \in [S]} \in [0,1]^{[S]}$ with $\sum_{\ell=1}^S y_{\ell} \leq W$, we have
\begin{align*}
\max_{\ell \in [S]} \frac{\ell+1}{W+\sum_{j=1}^{\ell} j \cdot y_j} \ge \frac{H_{S}}{2W}.
\end{align*}
\end{lemma}
\begin{proof}
Suppose for contradiction that there exist $S \in \bbN$, $W \ge 1$, and $\yv=(y_{\ell})_{\ell\in [S]} \in [0,1]^{[S]}$ such that $\sum_{\ell=1}^S y_{\ell} \leq W $ and, for every $\ell \in [S]$, 
\begin{align*}
     \frac{\ell+1}{W+\sum_{j=1}^\ell j \cdot y_j} < \frac{H_{S}}{2W} \Rightarrow {\textstyle\sum_{j=1}^\ell} j \cdot y_j > \frac{2W \cdot (\ell+1)}{H_S}-W.
\end{align*}

Dividing the above equation by $\ell \cdot (\ell+1)$ and summing over $\ell \in [S]$, we have that
\begin{equation}\label{eqn:lem-lower-bound}
\sum_{\ell=1}^S \sum_{j=1}^\ell \frac{j \cdot y_j}{\ell \cdot (\ell+1)} > \sum_{\ell=1}^S \left( \frac{2W}{H_S \cdot \ell} - \frac{W}{\ell \cdot (\ell+1)} \right).
\end{equation}

Let us analyze the LHS in \Cref{eqn:lem-lower-bound} by exchanging the order of summations. We have
\begin{align*}
    \text{LHS} &= \sum_{j=1}^S j \cdot y_j \cdot \left(\sum_{\ell=j}^S \frac{1}{\ell \cdot (\ell+1)} \right)
    = \sum_{j=1}^S j \cdot y_j \cdot \left(\sum_{\ell=j}^S \frac{1}{\ell} - \frac{1}{\ell+1} \right)\\
    &= \sum_{j=1}^S j \cdot y_j \cdot \left(\frac{1}{j} - \frac{1}{S+1}\right)
    \le \sum_{j=1}^S y_j \le W,
\end{align*}
where the third transition holds due to the telescopic sum. 

Now, let us analyze the RHS in \Cref{eqn:lem-lower-bound} using the same telescopic sum. 
\begin{align*}
    \text{RHS} = \frac{2W}{H_S} \sum_{\ell=1}^S \frac{1}{\ell} - W \cdot \sum_{\ell=1}^S \left(\frac{1}{\ell}-\frac{1}{\ell+1}\right)
    = \frac{2W}{H_S} \cdot H_S - W \cdot \left(1-\frac{1}{S+1}\right)
    \ge W.
\end{align*}

Hence, we proved that in \Cref{eqn:lem-lower-bound}, the LHS is at most $W$ and the RHS is at least $W$. But the equation proves the LHS to be strictly greater than the RHS, which is the desired contradiction.  
%
%
\end{proof}

\begin{figure}[t]
    \centering
\begin{tabular}{|c| c c c c| c| c c c c| c c c c|  c c c c| }
\hline
\diagbox[width=3em]{\textbf{$\bv_t$}}{$t$}&
 $1$ & $2$ & $3$ & $4$ & $5-8$  & $9$ & $10$ & $11$ & $12$& \begin{tabular}[x]{@{}c@{}}  $13$ \\ $17$ \end{tabular} &\begin{tabular}[x]{@{}c@{}}  $14$ \\ $18$ \end{tabular}  & \begin{tabular}[x]{@{}c@{}}  $15$ \\ $19$ \end{tabular} & \begin{tabular}[x]{@{}c@{}}  $16$ \\ $20$ \end{tabular} 
 & \begin{tabular}[x]{@{}c@{}}  $21$ \\ $25$ \\ $29$\\$33$ \end{tabular}&   \begin{tabular}[x]{@{}c@{}}  $22$ \\ $26$ \\ $30$\\$34$ \end{tabular} &  \begin{tabular}[x]{@{}c@{}}  $23$ \\ $27$ \\ $31$\\$35$ \end{tabular} &  \begin{tabular}[x]{@{}c@{}}  $24$ \\ $28$ \\ $32$\\$36$ \end{tabular}  \\
  \hline
 $v_{1,t}$ & $1$ & $0$ & $0$ & $1$ & $0$ & $1$ & $0$ & $0$ & $0$& $1$ & $0$ & $0$ & $0$& $1$ & $0$ & $0$ & $0$ \\
 $v_{2,t}$ & $1$ & $1$ & $0$ & $0$ & $0$ & $0$ & $1$ & $0$ & $0$& $0$ & $1$ & $0$ & $0$ & $0$ & $1$ & $0$ & $0$  \\
  $v_{3,t}$ & $0$ & $1$ & $1$  & $0$ & $0$ & $0$ & $0$ & $1$ & $0$& $0$ & $0$ & $1$ & $0$& $0$ & $0$ & $1$ & $0$ \\
  $v_{4,t}$ & $0$ & $0$ & $1$  & $1$ & $0$ & $0$ & $0$ & $0$ & $1$& $0$ & $0$ & $0$ & $1$&$0$ & $0$ & $0$ & $1$\\
  \hline
  \multicolumn{1}{c}{} 
  & \multicolumn{4}{c}{\upbracefill}
  & \multicolumn{1}{c}{\upbracefill}
  & \multicolumn{4}{c}{\upbracefill}
  & \multicolumn{4}{c}{\upbracefill}
  &\multicolumn{4}{c}{\upbracefill}
   \\[-1ex]
  \multicolumn{1}{c}{} 
  & \multicolumn{4}{c}{$\scriptstyle S_2$}
  & \multicolumn{1}{c}{$\scriptstyle S_0$}
  & \multicolumn{4}{c}{$\scriptstyle  S_1$}
  & \multicolumn{4}{c}{$\scriptstyle 2\times S_1$}
  & \multicolumn{4}{c}{$\scriptstyle 4\times S_1$}\\
  \multicolumn{1}{c}{} 
  & \multicolumn{5}{c}{\upbracefill}&
   \multicolumn{8}{c}{\upbracefill}&
   \multicolumn{4}{c}{\upbracefill}\\[-1ex]
  \multicolumn{1}{c}{}
  & \multicolumn{5}{c}{$\scriptstyle L_1$}
  & \multicolumn{8}{c}{$\scriptstyle L'_2$}& \multicolumn{4}{c}{$\scriptstyle L'_3$}\\
\end{tabular}
\vspace{1cm}

\begin{tabular}{|c| c c c c| c| c c c c| c c c c|  c | c | }
\hline
\diagbox[width=3em]{\textbf{$\bv_t$}}{$t$}&
 $1$ & $2$ & $3$ & $4$ & $5-8$  & $9$ & $10$ & $11$ & $12$& \begin{tabular}[x]{@{}c@{}}  $13$ \\ $17$ \end{tabular} &\begin{tabular}[x]{@{}c@{}}  $14$ \\ $18$ \end{tabular}  & \begin{tabular}[x]{@{}c@{}}  $15$ \\ $19$ \end{tabular} & \begin{tabular}[x]{@{}c@{}}  $16$ \\ $20$ \end{tabular} 
 & $21-24$&  $25-36$  \\
  \hline
 $v_{1,t}$ & $1$ & $0$ & $0$ & $1$ & $0$ & $1$ & $0$ & $1$ & $1$& $0$ & $0$ & $0$ & $0$& $1$ & $0$  \\
 $v_{2,t}$ & $1$ & $1$ & $0$ & $0$ & $0$ & $1$ & $1$ & $0$ & $1$& $0$ & $0$ & $0$ & $0$ & $1$ & $0$   \\
  $v_{3,t}$ & $0$ & $1$ & $1$  & $0$ & $0$ & $1$ & $1$ & $1$ & $0$& $0$ & $0$ & $0$ & $0$& $1$ & $0$  \\
  $v_{4,t}$ & $0$ & $0$ & $1$  & $1$ & $0$ & $0$ & $1$ & $1$ & $1$& $0$ & $0$ & $0$ & $0$ &$1$ & $0$ \\
  \hline
  \multicolumn{1}{c}{} 
  & \multicolumn{4}{c}{\upbracefill}
  & \multicolumn{1}{c}{\upbracefill}
  & \multicolumn{4}{c}{\upbracefill}
  & \multicolumn{4}{c}{\upbracefill}
  &\multicolumn{1}{c}{\upbracefill}
   &\multicolumn{1}{c}{\upbracefill}
   \\[-1ex]
  \multicolumn{1}{c}{} 
  & \multicolumn{4}{c}{$\scriptstyle S_2$}
  & \multicolumn{1}{c}{$\scriptstyle S_0$}
  & \multicolumn{4}{c}{$\scriptstyle  S_3$}
  & \multicolumn{4}{c}{$\scriptstyle 2\times S_1$}
  & \multicolumn{1}{c}{$\scriptstyle S_4$}
    & \multicolumn{1}{c}{$\scriptstyle 3\times S_0$}\\
  \multicolumn{1}{c}{} 
  & \multicolumn{5}{c}{\upbracefill}&
   \multicolumn{8}{c}{\upbracefill}&
   \multicolumn{2}{c}{\upbracefill}\\[-1ex]
  \multicolumn{1}{c}{}
  & \multicolumn{5}{c}{$\scriptstyle L_1$}
  & \multicolumn{8}{c}{$\scriptstyle L_2$}
  & \multicolumn{2}{c}{$\scriptstyle L_3$}\\
\end{tabular}
    \caption{Instances $I_1$ (top) and $I_3$ (bottom) when $N=4$.}
    \label{fig:Instances_binary}
\end{figure}

Now, we are ready to prove \Cref{thm:lower-bound-binary}.
\begin{proof}[Proof of \Cref{thm:lower-bound-binary}]
We prove that the Nash social welfare approximation of every online algorithm is at least $H_N/2 = \Omega(\log N)$. Suppose for contradiction that there is an online algorithm with a smaller approximation ratio. 

Set $T= N \cdot (N^2+N-2)/2$. We use $\bv_{t}=(v_{i,t})_{i\in [N]}$ to denote the vector of values of all the agents for good $t \in [T]$. For $r \in [N] \cup \set{0}$, we denote with $S_r$ a sequence of $N$ goods that arrive consecutively for which the agents have the following valuations. For the first good in the sequence, $v_{i,1} = 1$ for $i \le r$ and $v_{i,1} = 0$ for $i \ge r+1$. For $t \in \set{2,\ldots,N}$, $\bv_t$ is obtained by cyclically permuting $\bv_{t-1}$ to the right. For example, if $N=3$, then $S_2$ consists of $3$ goods that arrive consecutively with $\bv_1=(1,1,0)$, $\bv_2=(0,1,1)$, and $\bv_3=(1,0,1)$. Notice that $S_r$ consists of $N$ goods such that each good is liked by $r$ agents and each agent likes $r$ goods. Also, notice that $S_0$ consists of $N$ goods for which all agents have zero value. 

Next, we use $S_r$ to construct two building blocks of our adversarial instances.
\begin{itemize}
    \item For $r \in [N-1]$, let $L_r$ be a sequence of consecutive goods constructed by concatenating $S_{r+1}$ followed by $r$ copies of $S_0$. 
    \item For $r \in [N-1]$, let $L'_r$ be a sequence of consecutive goods constructed by concatenating $r+1$ copies of $S_1$.
\end{itemize}

Notice that for each $r\in [N-1]$, $L_r$ and $L'_r$ both consist of $(r+1) \cdot N$ goods of which each agent likes exactly $r+1$ goods. 

Finally, for each $k \in [N-1]$, define instance  $I_k=(L_1,L_2,\ldots,L_k,L'_{k+1},\ldots,L'_{N-1})$, i.e., instance $I_k$ is constructed by concatenating $L_1, L_2, \ldots, L_k, L'_{k+1}, \ldots, L'_{N-1}$ in that order. \Cref{fig:Instances_binary} shows instances $I_1$ and $I_3$ with $N=4$. Notice that for each $k \in [N-1]$, instance $I_k$ consists of a total of $T$ goods of which each agent likes $(N^2+N-2)/2$ goods. We let the algorithm know in advance that it will see instance $I_k$ for some $k \in [N-1]$, and prove that it still cannot achieve proportional fairness better than $H_N/2$. 

With slight abuse of notation, for $r \in [N-1]$, let $L_r$ also denote the set of goods that appear in the sequence $L_r$. Let $\xv$ denote the allocation produced by the online algorithm on instance $I_{N-1}$, i.e., when the algorithm observes $L_1,L_2,\ldots,L_{N-1}$. For $k \in [N-1]$, let $x_{L_k} = \sum_{t \in L_k} x_t$ denote the total allocation to goods in $L_k$. 

Now, fix $k \in [N-1]$ and let $\xv^k$ denote the allocation produced by the algorithm on instance $I_k$ with values $v_{i,t}$. Because the algorithm cannot distinguish between $I_k$ and $I_{N-1}$ for the first $k$ blocks (namely, $L_1,\ldots,L_k$), even under instance $I_k$, the algorithm must assign a total allocation of $x_{L_\ell}$ to block $L_{\ell}$ for each $\ell \le k$. Hence, we have 
\allowdisplaybreaks
\begin{align*}
  \NSW(\xv^k) &=  \left( \prod_{i=1}^N \left(
    \sum_{\ell=1}^k \sum_{ t \in L_{\ell}} v_{i,t} \cdot x^k_t   +  \sum_{\ell=k+1}^{N-1}  \sum_{ t \in L'_{\ell}} v_{i,t} \cdot x^k_t
    \right) \right)^{1/N} \\
&    \leq \frac{1}{N}\cdot \left(   \sum_{i=1}^N  \left( \sum_{\ell=1}^k \sum_{ t \in L_{\ell}} v_{i,t} \cdot x^k_t   +  \sum_{\ell=k+1}^{N-1}  \sum_{ t \in L'_{\ell}} v_{i,t} \cdot x^k_t 
    \right) \right)\\
 &   = \frac{1}{N}\cdot \left(     \sum_{\ell=1}^k \sum_{ t \in L_{\ell}} \sum_{i=1}^N  v_{i,t} \cdot x^k_t   +  \sum_{\ell=k+1}^{N-1}  \sum_{ t \in L'_{\ell}} \sum_{i=1}^N  v_{i,t} \cdot x^k_t 
    \right) \\
     &\leq \frac{1}{N}\cdot \left(     \sum_{\ell=1}^k \sum_{ t \in L_{\ell}} (\ell+1)  \cdot x^k_t   +  \sum_{\ell=k+1}^{N-1}  \sum_{ j \in L'_{\ell}} x^k_t
    \right) \\
     & \leq  \frac{1}{N}\cdot \left(     \sum_{\ell=1}^k  (\ell+1)  \cdot x_{L_\ell}   +  1-\sum_{\ell=1}^k x_{L_\ell}    \right) 
\end{align*}
where the second transition follows from the AM-GM inequality and the fourth transition follows because, for each $\ell \in [N-1]$, each good $t \in L_{\ell}$ is liked by at most $\ell+1$ agents 
and each good $t \in L'_{\ell}$ is liked by a single agent. 

Consider an alternative allocation $\yv^k$ which allocates $\nicefrac{1}{N}$ to each of the first $N$ goods of $L_k$ (i.e., goods of its $S_{k+1}$ portion). This provides each agent utility equal to $(k+1)/N$, thus achieving Nash social welfare equal to $(k+1)/N$. 

Hence, the algorithm's approximation ratio $\alpha_k$ on instance $I_k$ satisfies
\begin{align*}
  \alpha_k & \geq \frac{k+1}{\sum_{\ell=1}^k  (\ell+1)  \cdot x_{L_\ell}   +  1-\sum_{\ell=1}^k x_{L_\ell}}
 \ge \frac{k+1}{1+\sum_{\ell=1}^k  (\ell+1)  \cdot x_{L_\ell}}.
\end{align*}

Hence, the worst-case approximation ratio is at least $\max_{k \in [N-1]} \alpha_k$. Applying \Cref{lem:lower-bounds} with $S=N$ and $W=1$, we get that this is at least $H_N/2$, as desired.\footnote{We can set $y_1 = 0$ and $y_{\ell} = x_{L_{\ell-1}}$ for all $\ell \in \set{2,\ldots,N}$ when applying the lemma.} 
\end{proof} 

Next, we provide a set-aside greedy algorithm that achieves $O(\log N)$ proportional fairness (and therefore, $O(\log N)$-NSW optimality), thus establishing $\Theta(\log N)$ as the best possible approximation in this restricted scenario. We remark that this restricted case of binary utilities and unit budget already poses an interesting challenge by preventing constant approximation, but $O(\log N)$ approximation is still quite reasonable as it does not depend on the horizon $T$ (which can typically be very large) and in practice the number of agents $N$ is reasonably small. We also remark that we achieve the $O(\log N)$ upper bound using a horizon-independent algorithm, while the lower bound of \Cref{thm:lower-bound-binary} holds even when the algorithm is horizon-aware. 

The algorithm (\Cref{alg:Binary}), works as follows: it uses the set-aside portion of the budget to set $y_t = 1/(2N)$ whenever good $t$ is the first liked good of at least one agent (note that $\sum_t y_t \le 1/2$). This ensures that each agent $i$ gets utility at least $\Delta_i = 1/(2N)$.\footnote{Agents who do not like any good can only improve the approximation.} Based on this, the algorithm uses the following expression of promised utility to agent $i$ in round $t$: 
$
\tilde{u}_{i,t}(z_t) =\frac{1}{2N} + {\textstyle\sum_{\tau=1}^t} v_{i,\tau} \cdot z_{\tau}.
$

\begin{algorithm}[!ht]
\caption{\algoname\ Algorithm for Binary Values and Unit Budget}
\DontPrintSemicolon
\KwIn{Target proportional-fairness level $\alpha$}
\begin{algorithmic}[1]
\label{alg:Binary}
\FORALL{$t=1$ to $T$}
\STATE (Set-aside semi-allocation) If there exists $i \in [N] $ with $v_{i,t} = 1$ and $v_{i,\tau} = 0$ for each $\tau \in [t-1]$, then set $y_t=1/(2N)$, else set $y_t = 0$.
\STATE (Greedy semi-allocation) Compute ${z}_t = \min\set{z_t :\frac{1}{N} \sum_{i=1}^N \frac{v_{it}}{\tilde{u}_{it}(z_t)} \leq \alpha, \; z_t \geq 0}$. 
\STATE Allocate ${x}_t = {y}_t+{z}_t$ to good $t$.
\ENDFOR
\end{algorithmic}
\end{algorithm}

The algorithm chooses $z_t$ in a greedy manner (i.e., smallest possible) such that, for each agent $i$, the ratio of her value $v_{i,t}$ for good $t$ to her promised utility $\tilde{u}_{i,t}(z_t)$ is at most a target quantity. We defer the proof of this result to the appendix because we will present this technique in much more generality in \cref{sec:batched_alg} (the only adjustment required in the proof of the next result is the slightly different expression of $\Delta_i = 1/(2N)$ specific to this case of binary utilities and unit budget).

\begin{theorem}\label{thm:binarypf}
\Cref{alg:Binary} with $\alpha\geq 2\ln(2N)$ realizes an $\alpha$-proportional fair allocation. 
\end{theorem}

\section{General Utilities and Budget}
\label{sec:single_public}
Having built some intuition about the online setting and the proportional fairness objective by considering the restricted special case of the problem wherein all values $v_{it}$ are in $\{0,1\}$ and the budget is $B=1$, we now turn to the more general model described in \Cref{sec:model}. Recall that this model generalizes the setting in \Cref{sec:binary} in two ways: 
\begin{enumerate}
    \item Agent values $v_{i,t}$ can now be any (non-negative) real number.
    \item The budget constraint is $\sum_{t=1}^T x_t \leq B$, for an arbitrary $B \ge 0$, so the per-round constraint of $x_t \le 1$, for each $t \in [T]$, is no longer redundant.   
\end{enumerate}

\subsection{The Case for Predictions}
\label{ssec:predictions}

In this general case, we prove that the problem becomes significantly more difficult: without access to any predictions, every online algorithm is $\Omega(T/B)$-proportionally fair (in fact, achieves $\Omega(T/B)$-approximation of the Nash welfare), in stark contrast to the $O(\log N)$-proportional fairness that we were able to achieve in the previous section. 

\begin{proposition}\label{prop:hardness_no_predictions}
Under general agent values and budget $B$, every online algorithm is $\Omega(T/B)$-proportionally fair (in fact, achieves $\Omega(T/B)$-approximation of the Nash welfare).
\end{proposition}

\begin{proof}
The hardness instances we consider will have $N = 1$ agent. In this case, the proportional fairness of an online algorithm on a given instance is equal to the maximum possible utility the agent could have obtained in hindsight divided by the utility she obtained under the algorithm; this is also equal to the approximation ratio for the Nash welfare. For clarity, we will omit the subscript $i$ in the notation in this proof, so that $v_t$ is the value that the agent has for good $t$, and $u(\xv)$ is the utility of the agent under allocation $\xv$. With $N=1$, a problem instance is defined by a sequence of $T$ values $(v_1, v_2, \ldots, v_T)$. 

Consider the following family of $T$ instances: for $t \in [T]$, instance $I_t := (1, M, \ldots, M^{t-1}, 0, \ldots, 0)$, where $M$ is a sufficiently large number. The algorithm knows in advance that it will see instance $I_t$ for some $t \in [T]$, and prove that it still cannot achieve proportional fairness better than $T/2B$. 

Let $\xv = (x_1,\ldots,x_T)$ be the allocation produced by the algorithm on instance $I_T = (1,M,\ldots,M^{T-1})$. For any $t \in [T]$, since the algorithm cannot distinguish between $I_t$ and $I_{T}$ up to round $t$, the allocation up to round $t$ under instance $I_t$ must also be $(x_1,\ldots,x_t)$. Further, without loss of generality, we can assume that the algorithm allocates $0$ in any round where the value is $0$. Therefore, for each $t \in [T]$, the allocation produced by the algorithm on instance $I_t$ is $(x_1, \ldots, x_t, 0, \ldots, 0)$. 

We claim that in order to be $\frac{T}{2B}$-proportionally fair, the algorithm needs to set $x_t \geq \frac{B}{T}$ for all $t$. To see this, suppose to the contrary that $x_t < \frac{B}{T}$ for some $t$. Then on instance $I_t$, the agent's utility under the algorithm is
$$u(\xv) 
= \sum_{\tau=1}^t x_\tau M^\tau 
< \sum_{\tau=1}^{t-1}M^\tau + \frac{B}{T} \cdot M^t < \frac{2B}{T}M^t,
$$
where the last inequality holds when $M$ is large enough. On the other hand, the hindsight-optimal allocation on $I_t$ is to simply set $x^*_t = 1$, which gets utility $M^t$. Therefore, if $x_t < \frac{B}{T}$ then the algorithm cannot be $\frac{T}{2B}$-proportionally fair on $I_t$.

We have shown that $x_t \geq \frac{B}{T}$ for all $t$ is necessary for the algorithm to be $\frac{T}{2B}$-proportionally fair. However, since the overall budget is $B$, the only way this can happen is if $x_t = \frac{B}{T}$ for all $t$. The same calculation as above shows that under this allocation, the algorithm is at most $(T/2B)$-proportionally fair, establishing the desired $\Omega(T/B)$ lower bound. 
\end{proof}

The hardness instance used above is specifically engineered to exploit the fact that the algorithm has no information about the future. In most practical settings, it is reasonable to assume that the algorithm has access to some information about the input. This could come from historical data, stochastic assumptions, or simply from properties of the problem at hand (e.g. if $v_{i,t}$ represents the monetary value that agent $i$ has for good $t$, then we may have bounds on how large $v_{i,t}$ can be.)

Motivated by this, we now turn to prediction-augmented algorithms and allow the algorithm access to additional side-information about agents' valuations. Clearly, if the entire valuation matrix $(v_{i,t})_{i\in[N],t\in[T]}$ is available beforehand, then the problem is trivial; the challenge lies in understanding what \emph{minimal} additional information (or `prediction') can lead to sharp improvements in performance, and how robust these improvements are to errors in these predictions. 
To this end, we now adapt an idea introduced by~\citet{BGGJ22} for online allocation with private goods, and assume that the algorithm has side information about each agent's total value for all the goods. 

\begin{definition}[Total Value Predictions]
\label{def:predictions}
For any agent $i$, we define her total value to be $V_i=\sum_{t=1}^Tv_{i,t}$. 
Moreover, for $c_i,d_i \ge 1$, $\tilde{V}_i$ is said to be a $[c_i,d_i]$-prediction of $V_i$ if $\tilde{V_i} \in \left[\frac{1}{d_i} V_i, ~ c_i V_i\right]$. 
\end{definition}
In other words, $c_i$ and $d_i$ denote the multiplicative factors by which the prediction $\tilde{V}_i$ may overestimate and underestimate, respectively, the value of $V_i$. When $c_i=d_i=1$, we call them \emph{perfect predictions}. 

In the next section, we assume that we have access to $\tilde{V}_i$ for each agent $i$. The purpose of the $c_i$ and $d_i$ is to parameterize the \emph{robustness} of our algorithm, i.e., the degradation in its performance as the predictions get worse. Our algorithm does need to know (an upper bound on) the $d_i$'s for tuning one of its parameters; it does not however need to know the $c_i$'s (these are only used in the analysis).

\subsection{Achieving Proportional Fairness with Predictions}
\label{sec:general_alg}

Using the above notion of predictions, we design \Cref{alg:General-L-1}, a variant of our earlier \algoname algorithm that has a dramatically better proportional fairness guarantee compared to the hardness result of $\Omega(T/B)$ in~\cref{prop:hardness_no_predictions}. Given perfect predictions, our algorithm achieves a proportional fairness of $O(\log(T/B))$. Moreover,~\cref{alg:General-L-1} turns out to be remarkably robust to prediction errors; in particular, all our asymptotic guarantees remain unchanged as long as all the $c_i=O(1)$ and $d_i=O(\log(T/B))$. 

As before, the algorithm splits the budget into two parts, and the total allocation to the good in round $t$ is obtained by adding the contributions (semi-allocations) from each part, i.e., $ x_{t} =  y_{t}+z_{t}$. The semi-allocation from the first (\emph{set-aside}) part is $y_t=B/(2T)$ for each $t \in [T]$. This portion of the allocation guarantees each agent $i$ utility at least $\Delta_i = \frac{B}{2T} \cdot V_i$. Now in each round $t\in[T]$, the algorithm uses the second the part of the budget to compute a \emph{greedy} semi-allocation $z_t$. This is done by choosing $z_t$ to optimize a function of the agents' \emph{predicted promised utilities}.


\begin{definition}[Predicted Promised Utility]
\label{def:predicted}
Given a prediction $\tilde{V}_i$ of the total value of agent $i$, the predicted promised utility of agent $i$ in round $t\in[T]$ is defined as
\begin{equation}
\label{eq:pred}
    \tilde{u}_{i,t}(z_t) = \frac{B}{2T} \cdot \tilde{V}_i + \sum_{\tau = 1}^t v_{i,\tau} \cdot z_{\tau}.
\end{equation}
We omit $z_1,\ldots,z_{t-1}$ from the argument of $\tilde{u}_{i,t}$ since they are fixed prior to round $t$. 
\end{definition}

Note that this quantity can be computed by the algorithm, as a function of $z_t$ it wants to choose, since it has knowledge of $\tilde{V}_i$ (prediction) and semi-allocations $\{z_{\tau}\}_{\tau<t}$ from the previous rounds. We use these predicted promised utilities in \Cref{alg:General-L-1} to achieve the following guarantee.


\begin{algorithm}[!ht]
\caption{\algoname\ algorithm for the non-batched setting.}
\KwIn{Target threshold $\alpha$; total value predictions $(\tilde{V}_i)_{i\in[N]}$}
\begin{algorithmic}[1]
\label{alg:General-L-1}
\FORALL{$t=1$ to $T$}
\STATE Set-aside semi-allocation: set ${y}_t = \frac{B}{2T}$. 
\STATE Greedy semi-allocation: ${z}_t = \min\{z^*_t, 1 - {y}_t \}$, where 
    $z^*_t = \min \left\{z \geq 0:  \frac{1}{N} \sum_{i=1}^N \frac{v_{i,t}}{\tilde{u}_{i,t}(z)}\leq \frac{\alpha}{2B}\right\}$. 
\STATE Allocate $ x_{t} =  y_{t}+ z_{t}$.
\ENDFOR
\end{algorithmic}
\end{algorithm}

\begin{restatable}{theorem}{general}
\label{thm:general}
For any $\alpha \geq 4\ln\left(\frac{2T}{B}\right) + \frac{4}{N} \sum_i \ln(d_i)$, \cref{alg:General-L-1} produces a feasible allocation $\xv$, which satisfies
$
\textstyle\max_{\wv} \frac{1}{N} \sum_{i=1}^N \frac{u_i(\wv)}{c_i \cdot u_i({\xv})} \leq \alpha,
$
where the maximum is over all feasible allocations $\wv$. 
\end{restatable}

This result is almost subsumed by our positive result (\Cref{thm:batched}) for the more general model in the next section, in which we allow a batch of $L$ public goods to arrive in each round; setting $L=1$ recovers precisely the bounds derived in this section. The only difference is that in the algorithm for the batched model (\Cref{alg:batched}), the greedy semi-allocation step is a convex optimization problem, which can be solved up to an $\epsilon$ error; for the single-good-per-round case, we are able to replace this with a combinatorial step in \Cref{alg:General-L-1} that can be performed exactly in polynomial time. This requires us to show that the greedy semi-allocations computed by this combinatorial step satisfy the same properties (specifically, those in \Cref{prop:properties_batched}) that a solution to the convex optimization problem would have satisfied. We present this result and its proof in the appendix as \Cref{prop:properties}. 

Let us consider the implications of \Cref{thm:general}. The expression in the statement of \Cref{thm:general} is not exactly the proportional fairness objective, since the term for each agent $i$ is scaled with a (potentially different) factor $c_i$. However, applying the arguments in \Cref{subsec:prop_fairness}, we can nonetheless convert this into an approximation of proportional fairness, the core, and the Nash social welfare. The proof is in the appendix. 
\begin{corollary}\label{cor:general_fair}
For $\alpha \geq 4\ln\left(\frac{2T}{B}\right) + \frac{4}{N} \sum_i \ln(d_i)$, \cref{alg:General-L-1} is 
\begin{enumerate}
    \item $(\alpha \cdot \max_{i\in [N]} c_i)$-proportionally fair, and hence in the $(\alpha \cdot \max_{i \in [N]} c_i)$-core, and
    \item achieves $(\alpha \cdot (\prod_{i \in [N]} c_i)^{\frac{1}{N}})$-approximation of the Nash welfare.
\end{enumerate}
\end{corollary}

\begin{proof}

(1) \Cref{thm:general} implies that 
    $$\frac{1}{\max_i c_i} \cdot \max_{\wv} \frac{1}{N} \sum_{i=1}^N \frac{u_i(\wv)}{u_i({\xv})} \leq \max_{\wv} \frac{1}{N} \sum_{i=1}^N \frac{u_i(\wv)}{c_iu_i({\xv})} \leq \alpha,$$
    which directly implies the desired proportional fairness guarantee and, by \Cref{prop:core}, the desired core guarantee. While the same guarantee carries over to Nash welfare approximation, repeating the proof of \Cref{prop:nsw} actually provides a better approximation. 
    
(2) Let $\xv^*$ be the allocation maximizing the Nash social welfare in hindsight. We have
$$\max_{\wv} \frac{1}{N} \sum_{i=1}^N \frac{u_i(\wv)}{c_iu_i({\xv})}
\geq \frac{1}{N} \sum_{i=1}^N \frac{u_i(\xv^*)}{c_iu_i({\xv})}
\stackrel{(a)} \geq \left(\prod_{i=1}^N \frac{u_i(\xv^*)}{c_iu_i({\xv})}\right)^{\frac1N}
\stackrel{(b)}= \left(\prod_{i=1}^N \frac{1}{c_i}\right)^{\frac1N} \frac{\NSW(\xv^*)}{{\NSW({\xv})}}$$
where $(a)$ is by the AM-GM inequality, and $(b)$ is by the definition of the Nash welfare. This, together with \Cref{thm:general}, yields the second part of the corollary.
\end{proof}

\subsection{Hardness with Predictions}
\label{ssec:publiclower}

\cref{thm:general} shows that in online allocation of public goods with general values and budget, having access to reasonable predictions of each agent's total value can lead to a dramatic improvement in the proportional fairness guarantee from $\Omega(T/B)$ to $O(\log(T/B)))$.  
Given the size of the side information (which lies in $\mathbb{R}^N$, since we need one prediction per agent) relative to the ambient size of the input (which lies in $\mathbb{R}^{NT}$, with one valuation per agent per round), this is a surprising improvement in performance.

One may wonder whether these predictions are so strong that one can do even better. The following result shows, however, that even with a single agent, and perfect knowledge of her total value $V_1$, any online algorithm is $\Omega(\log(T/B))$-proportionally fair. Thus~\cref{alg:General-L-1} is essentially optimal for our setting. 
\begin{theorem}\label{thm:general-hardness-with-predictions}
For $N=1$ agent, every online algorithm is $\Omega(\log(T/B))$-proportionally fair (in fact, achieves an $\Omega(\log(T/B))$-approximation for the Nash welfare), even with perfect knowledge of horizon $T$ and the total value of the agent $\tilde{V}_1 = V_1 = \sum_{t=1}^T v_{i,t}$.
\end{theorem}

\begin{proof}
As in the proof of \Cref{prop:hardness_no_predictions}, approximations to proportional fairness and Nash welfare are equivalent with $N=1$ agent, both coinciding with the ratio of the agent's maximum utility in hindsight to the agent's utility under the algorithm. 

For the sake of contradiction, suppose that that there exists an online algorithm whose approximation ratio is $o(\log(T/B))$.
Let $T= B\cdot (T'(T'+1)-2) /2$ for some $T'$. For $r \in[T'-1]$, we denote with $S_r$ the sequence of $B\cdot (r+1)$ goods that arrive consecutively for which the agent has utility equal to $(r+1)/T'$ for the first $B$ goods and utility equal to $0$ for the remaining goods and with $S'_r$ the sequence of $B\cdot (r+1)$ goods such that the agent has utility equal to $1/T'$ for each of the goods. Now, for each $k \in [T'-1]$, let $I_k=(S_1,S_2, \ldots, S_k,S'_{k+1}, \ldots, S'_{T'-1})$, i.e. $I_k$ be the instance that is constructed by concatenating  $S_1, S_2, \ldots S_k, S'_{k+1},\ldots, S_{T'-1}$ in that order.  
Notice that for $r \in [T'-1]$, $S_r$ and $S'_r$ have the same number of goods, equal to $B(r+1)$, and for each of them the agent has the same accumulated utility, equal to $B(r+1)/T'$. Hence,  each instance $I_k$ consists of $T$ goods and the agent has the same accumulated utility for all these instances.  We assume that the algorithms is aware that  it will see instance $I_k$ for some $k\in [T'-1]$.

With slight abuse of notation, for $r \in [T'-1]$, we let $S_r$ to also denote the set of goods that appear in the sequence $S_r$. We   denote with $\xv$ the allocation produced by the online algorithm on instance $I_{T'-1}$. For $r \in [T'-1]$, let $x_{S_r} = \sum_{t \in S_r} x_t$ denote the total allocation to goods in $S_r$. Now, for $k \in [T'-1]$, let $\xv^k$ denote the allocation produced by the algorithm on instance $I_k$. As the algorithm cannot distinguish between $I_k$ and $I_{T'-1}$ for the first $k$ blocks (i.e. $S_1,\ldots,S_k$),  under instance $I_k$, the algorithm must assign a total allocation of $x_{S_\ell}$ to block $S_{\ell}$ for each $\ell \le k$.

 Moreover, under instance $I_k$, the optimal algorithm allocates  $1$ to each of the first $B$ goods of $S_k$.   Hence, we have that under instance  $I_k$, if $\alpha_k$ is the approximation ratio of the algorithm, then
\begin{align*}
  \alpha_k & \geq   \frac{B \cdot (k+1)/N}{ ( 2\cdot x_{S_1} + 3 \cdot x_{S_2}+\ldots +  (k+1) x_{S_{k}} +   (1-x_{S_1}-x_{S_2}-\ldots -x_{S_{k}} )  )/N }\\
 & =B \cdot \frac{ (k+1)}{  x_{S_1} + 2 \cdot x_{S_2}+\ldots +  k x_{S_k} +   1 }.
\end{align*}

Hence, if $\alpha$ is the worst-case approximation ratio over all the instances $I_k$, we have
\begin{align*}
    \alpha \geq  B \cdot \max_{k\in [T'-1]}\frac{(k+1)}{  x_{S_1} + 2 \cdot x_{S_2}+\ldots +  k x_{S_k} +   1 } 
\end{align*}
and by applying \Cref{lem:lower-bounds}, with $S=T'$ and $W=B$, we get that $\alpha \geq B \cdot H_{T'}/(2B)$. Thus, we have that  $\alpha = \Omega( \log(T')) = \Omega(\log(T/B))$. 
\end{proof} 

\section{Batched Arrival of Public Goods}
\label{sec:batched}

In this section, we present our most general setting, which we refer to as the \emph{batched public goods} model. This not only generalizes the model in \Cref{sec:model} with a single good per round, but also the setting of \citet{BGGJ22} with private goods. 

In each round $t \in [T]$, a batch of $L$ public goods arrive (as opposed to a single public good). Upon their arrival, the algorithm learns the value $v_{i,l,t} \ge 0$ of each agent $i \in [N]$ for each good $l \in [L]$ in the batch. It then must irrevocably decide the allocation $x_{l,t} \in [0,1]$ to each good $l$ in the batch, before the next round. We use $\xv_t = (x_{l,t})_{l \in [L]}$ to denote the allocation in round $t$, and $\xv = (\xv_t)_{t \in [T]}$ to denote the final allocation. We also incorporate two types of constraints on the allocation $\xv$:
\begin{enumerate}
	\item (Per-round constraint) $\sum_{l=1}^L x_{l,t} \leq 1$ for all $t \in [T]$. 
	\item (Overall constraint) $\sum_{t=1}^T \sum_{l=1}^L x_{l,t} \leq B$ for $B \geq 0$ known to the algorithm in advance.
\end{enumerate}
The allocation $x_{l,t}$ to good $l$ in round $t$ yields utility $v_{i,l,t} \cdot x_{l,t}$ to every agent $i$. We assume that agent utilities are additive, i.e., the final utility of agent $i$ is given by $u_i(\xv) = \sum_{t=1}^T \sum_{l=1}^L v_{i,l,t} \cdot x_{l,t}$. Our aim, as before, is is to realize an $\alpha$-proportionally fair allocation (\Cref{def:pf} does not require any modifications except using these new utility functions) for the smallest possible $\alpha$.

Note that the overall constraint is the same as in the model in \Cref{sec:model} with a single good per round. However, in the batched model, the per-round constraint can place additional restriction on how much budget can be spent in any single round. Note also that choosing the per-round bound to be $1$ is without loss of generality; in particular, since agent utilities are linear, having a per-round constraint of $b$ can be reduced to our setting by scaling each allocation, as well as the total budget, by a factor of $b$. Finally, the per-round constraint becomes vacuous if $B \leq 1$, and the overall constraint becomes vacuous if $B \geq T$; therefore, we assume, without loss of generality, that $1 \leq B \leq T$. 

\emph{This model captures the following special cases.} Before we dive into the algorithm and analysis, we briefly mention some special cases of interest which this model generalizes.

\begin{enumerate}
    \item \emph{Single public good.} When $L = 1$, we trivially recover the setting of \Cref{sec:single_public}, where there is one public good in each round. 
    \item \emph{Single private good.} When $L = N$, $B=T$, and $v_{i,l,t} = 0$ if $i \neq j$, we recover the setting studied by \citet{BGGJ22}. In their setting, there is a single private good arriving in each round, which the algorithm needs to split among the $N$ agents. When cast in our model, $v_{i,i,t}$ is the value that agent $i$ has for the good in round $t$, and $x_{i,t}$ is the fraction of good $t$ that agent $i$ is allocated. Note that Banerjee et al. only study per-round constraints, so one of our contributions is a generalization of their result to the budgeted setting. 
    \item \emph{Batched private goods.} When $L = L' \cdot N$, $B=T$, and $v_{i,l,t} = 0$ if $i \not\equiv j \, (\mathrm{mod} \; N)$, we capture a setting where there are $L'$ private goods arriving in each round, and the algorithm can (fractionally) allocate at most $1$ good in total among the agents in each round. 
\end{enumerate}

\subsection{The \algoname Algorithm for Batched Public Goods}\label{sec:batched_alg}

We present an algorithm, \Cref{alg:batched}, for the batched public goods model which generalizes our guarantees from \Cref{sec:single_public} and, partly, from the work of \citet{BGGJ22} as well. First, we need to extend predicted promised utilities to the batched model. 

\begin{definition}[Predicted Promised Utility]
\label{def:predictedbatch}
The predicted promised utility of agent $i$ in round $t$ is
\begin{equation}
\label{eq:batchpred}
    \tilde{u}_{i,t}(\zv_1, \ldots, \zv_t) = \frac{B}{2\min\{N, L\}T} \cdot \tilde{V}_i + \sum_{\tau = 1}^t \sum_{l=1}^L v_{i,l,\tau} z_{l,\tau}
\end{equation}
For clarity, we will often omit the dependence on $\zv_1, \ldots, \zv_{t-1}$ and just write $\tilde{u}_{i,t}(\zv_t)$.
\end{definition}

\begin{algorithm}[ht]
\caption{\algoname Algorithm for Batched Public Goods}
\KwIn{Target proportional-fairness level $\alpha$; total value predictions $(\tilde{V}_i)_{i\in[N]}$}
\begin{algorithmic}[1]
\label{alg:batched}
\STATE Define $q = \frac{\alpha}{2B}$. 
\FORALL{$t=1$ to $T$}
\STATE Given values $\{v_{i,l,t}\}_{i\in[N],l\in[L]}$, find `favorite goods' set $F_t$ as follows: Initialize $F_t \gets \emptyset$. 
\FORALL{$i=1$ to $N$}
    \STATE Let $l_i \gets \arg\max_l \{ v_{i,l,t} \}$ (breaking ties arbitrarily), and update  $F_t \gets F_t \cup \{l_i\}$.
\ENDFOR
\STATE Set-aside semi-allocation: set ${y}_{l,t} = \frac{B}{2\abs{F_t}T} \cdot \one\{l \in F_t\}$. 
\STATE Greedy semi-allocation: Let $(\zv_t, \lambda_t)$ be an optimal solution to the following optimization problem:\\
Maximize $\frac{1}{N} \sum_{i=1}^N \ln\left(\tilde{u}_{i,t}(\zv) \right) + \lambda q$ subject to $\sum_{l=1}^L z_{l} + \lambda = 1 - \frac{B}{2T}$ and $\zv, \lambda \geq 0$.
\STATE Allocate $x_{l,t} =  y_{l,t}+ z_{l, t}$.
\ENDFOR
\end{algorithmic}
\end{algorithm}

At a high level, \Cref{alg:batched} takes in as input a target approximation factor $\alpha$ (which we set later) and proceeds as follows. In round $t$, let $F_t$ be the set of `favorite goods', which satisfies that for each agent $i$, there exists some good $l \in F_t$ such that $l$ is among the most preferred goods of $i$ in round $t$. Note that this guarantees that $\abs{F_t} \leq \min\{N, L\}$. The algorithm allocates ${x}_{l,t} = {y}_{l,t} + {z}_{l,t}$, where the set-aside allocation is ${y}_{l,t} = \frac{B}{2\abs{F_t}T} \cdot \one\{l \in F_t\}$ for each $l \in [L]$ and the greedy allocation $\zv_t$ is computed by an optimization problem. Note that the objective function maximized is concave and the constraints are linear. Hence, this is a convex optimization problem, which can be solved up to any $\epsilon > 0$ accuracy in time that is polynomial in the input size and $1/\epsilon$. For simplicity, we present our proof assuming that the optimization problem is solved exactly, but solving it up to a sufficiently small error does not change our proportional fairness guarantee asymptotically. 

The purpose of the $\yv_t$ part of the allocation is to guarantee each agent a minimum share of their total utility $V_i$: Note that since $\abs{F_t} \leq \min\{N, L\}$, every agent $i$ is guaranteed to receive at least $\max_{l}\{v_{i,l,t}\}\cdot \frac{B}{2\min\{N,L\}T}$ utility from the $\yv_t$ set-aside semi-allocation in each round $t$. Summing over all $T$ rounds, each agent is guaranteed to receive at least $V_i \cdot \frac{B}{2\min\{N,L\}T}$ from the set-aside semi-allocation. \Cref{alg:batched} achieves the following guarantee; its proof is in \cref{sec:main_proof}.

\begin{restatable}{theorem}{batched}
\label{thm:batched}
In the batched public goods model, for any $\alpha\geq 4\ln\left(\frac{2\min\{N,L\}T}{B}\right) + \frac{4}{N} \sum_i \ln(d_i)$, \Cref{alg:batched} produces a feasible allocation $\xv$, which satisfies
$
\textstyle\max_{\wv} \frac{1}{N} \sum_{i=1}^N \frac{u_i(\wv)}{c_i \cdot u_i(\xv)} \le \alpha,
$
where the maximum is over all feasible allocations $\wv$. 
\end{restatable}

Using exactly the same proof as that of \Cref{cor:general_fair}, we obtain the following guarantees for \Cref{alg:batched} with respect to proportional fairness, the core, and the Nash social welfare.
\begin{corollary}\label{cor:batched_fair}
For $\alpha \geq 4\ln\left(\frac{2\min\{N,L\}T}{B}\right) + \frac{4}{N} \sum_i \ln(d_i)$, \Cref{alg:batched} is
\begin{enumerate}
    \item $(\alpha \cdot \max_i c_i)$-proportionally fair, and hence in the $(\alpha \cdot \max_i c_i)$-core.
    \item achieves $(\alpha \cdot (\prod_{i\in [N]} c_i)^{\frac{1}{N}})$-approximation of the Nash social welfare.
\end{enumerate}
\end{corollary}

Recall that the private goods setting of \citet{BGGJ22} is a special case of our setting in which $L = N$, $B=T$, and the valuation matrix of the $N$ agents for the $N$ public goods in each round is a diagonal matrix. For this special case, our Nash welfare approximation is $O((\prod_{i \in [N]} c_i)^{1/N} \cdot (\ln N + (1/N) \cdot \sum_i \ln d_i))$. The Nash welfare approximation obtained by \citet{BGGJ22} is almost the same, except that $\ln N$ is replaced by $\min\set{\ln N, \ln T}$. Thus, for $T = \Omega(N)$, our result strictly generalizes theirs and they prove this bound to be almost tight. For $T = o(N)$, they derive a better approximation that depends on $\ln T$ instead of $\ln N$, (and it is unknown if this is tight). It would be interesting to see if our result can also be improved in this case. 

\subsection{Proof of Theorem~\ref{thm:batched}}
\label{sec:main_proof}

Let $\xv$ denote the final allocation produced by \Cref{alg:batched}. Central to our analysis is the following linear program (and its dual), which essentially tries to \emph{maximize} (i.e., find the worst case) the proportional-fairness level achieved by~
\cref{alg:batched}:
\begin{equation*}
\begin{aligned}[t]
(P)\;
\max_{\wv\in\mathbb{R}^{LT}}\ &\frac{1}{N} \sum_{i=1}^N \frac{u_i(\wv)}{c_iu_i({\xv})} \\
\text{s.t.}\ &\sum_{t=1}^T\sum_{l=1}^L w_{l,t} \leq B\\
& \sum_{l=1}^L w_{l,t} \leq 1 \quad \forall \,t \in [T] \\
& w_{l,t} \geq 0 \quad \forall \, l \in [L], \, t \in [T]
\end{aligned}
\begin{aligned}[t]
(D)\; \min_{p\in \mathbb{R}^T, \, q \in \mathbb{R}}\ &\sum_{t=1}^T p_t + Bq \\
\text{s.t.}\ &p_t + q \geq \frac{1}{N} \sum_{i=1}^N \frac{v_{i,l,t}}{c_iu_i({\xv})}\quad\; \forall\, l \in [L], \, t\in[T]\\
&p_t, q \geq 0
\end{aligned}
\end{equation*}
Note that the dual constraints can be written more compactly as
$$p_t + q \geq \frac{1}{N} \max_l \sum_{i=1}^N \frac{v_{i,l,t}}{c_iu_i({\xv})}\quad\; \forall \, t\in[T].$$
From now on we will work with the dual constraints in the form above. 

\cref{alg:batched} can now be viewed as choosing $\xv$ in a manner such that it forces the LP $(P)$ to have as small a value as possible. One way to achieve this is to consider a `target' dual solution $q=\frac{\alpha}{2B}$ and $p_t = \max\{0, \frac{1}{N} \max_l \sum_{i=1}^N \frac{v_{i,l,t}}{{u}_{i}(\xv)} - q\}$ (for some appropriate choice of $\alpha$). To realize this,~\cref{alg:batched} wants to set the vector of allocations ${\xv}_t$ for each round $t$, as well as the corresponding dual variable $p_t$, in a manner such that the primal and dual solutions are consistent (i.e., all constraints are feasible, and they obey complementary slackness). 

The challenge in generating allocation ${\xv}$ and dual certificate $(p_t, q)$ in an \emph{online} fashion is that at time $t$, the  dual constraint depends on $u_i(\xv)$, the final utilities under the \emph{entire} allocation made by the algorithm. 
However, these quantities are not known to the algorithm at time $t$, because it cannot look into the future. To work around this, we use the predicted promised utilities (\cref{def:predictedbatch}) as proxy in place of $u_i(\xv)$ when making the decision in round $t$. 

The next lemma explains the sense in which the predicted promised utilities serve as a proxy; in particular, we show that under \Cref{alg:batched}, the predicted promised utility $\tilde{u}_{i,t}(\zv_t)$ is a lower bound on the true final utility of agent $i$ under $\xv$ up to a multiplicative factor of $c_i$, for each agent $i$ and in each round $t$.
\begin{lemma}
\label{lem:predicted_utilities}
For every agent $i\in[N]$ and time $t\in[T]$, we have
$\tilde{u}_{i,t}({\zv}_t) \leq c_i \cdot u_i({\xv}).$
\end{lemma}
\begin{proof}
We have
\begin{align*}
    u_i(\xv)
    &= \sum_{\tau=1}^T \sum_{l=1}^L v_{i,l,\tau}({y}_{l,\tau} + {z}_{l,\tau}) \\
    &= \sum_{\tau=1}^T \sum_{l=1}^L v_{i,l,\tau}{y}_{l,\tau} + \sum_{\tau=1}^T \sum_{l=1}^L v_{i,l,\tau}{z}_{l,\tau}  
\end{align*}
We next bound the contribution due to the set-aside allocations $\yv$. At any time $t\in[T]$, agent $i$ receives at least $\max_{l}\{v_{i,l,t}\} \, {y_{l,t}} = \max_{l}\{v_{i,l,t}\} \, \frac{B}{2\abs{F_t}T} $ utility from the set-aside semi-allocation $y_{l,t}$. Since $\abs{F_t} \leq \min\{N, L\}$, we get
$$\sum_{\tau=1}^T \sum_{l=1}^L v_{i,l,\tau}{y}_{l,\tau} 
\geq
\sum_{\tau=1}^T \max_{l}\{v_{i,l,t}\} \cdot \frac{B}{2\min\{N,L\}T}
\geq \frac{\tilde{V}_i}{c_i} \cdot \frac{B}{2\min\{N,L\}T}.
$$
Substituting this back in the above expression, and using the fact that $c_i\geq 1$ for all $i$, we have
\begin{align*}
    u_i(\xv)
    &\geq 
    \frac{\tilde{V}_i}{c_i} \cdot \frac{B}{2\min\{N,L\}T} + \sum_{\tau=1}^t \sum_{l=1}^L v_{i,l,\tau}{z}_{l,\tau}\\
    &\geq \frac{1}{c_i} \left( \frac{B}{2\min\{N, L\}T} \, \tilde{V}_i + \sum_{\tau = 1}^t \sum_{l=1}^L v_{i,l,\tau} z_{l,\tau}\right) 
    = \frac{1}{c_i} \, \tilde{u}_{i,t}(\zv_t).
\end{align*}
\end{proof}

Using this, we now formally define our new dual certificates, and list some key properties of the allocation and certificate which we need for our performance guarantee. Define 
\[
\Phi_l(\zv)= \frac{1}{N} \cdot \sum_{i=1}^N \frac{v_{i,l,t}}{\tilde{u}_{i,t}(\zv)} =\sum_{i=1}^N \frac{v_{i,l,t}}{\gamma_i+ \sum_{l=1}^L v_{i,l,t} z_{l,t}},
\]
where $\gamma_i=\frac{B\tilde{V}_i}{2\min\{N, L\}T} + \sum_{\tau = 1}^{t-1} \sum_{l=1}^L v_{i,l,\tau} z_{l,\tau}$; for the last transition, see the definition of predicted promised utilities (\Cref{eq:batchpred}).

\begin{proposition}[Some Properties of the Algorithm]
\label{prop:properties_batched}
Consider the dual certificate:
\begin{align*}
q &= \frac{\alpha}{2B} \geq \frac{2}{B} \left[\ln\left(\frac{2\min\{N, L\}T}{B}\right) + \frac{1}{N} \sum_i \ln(d_i)\right]\\
p_t &= \max\left\{0, \max_l \Phi_l(\zv_t) - q\right\}
\end{align*}
Then, the allocation $\xv = \yv+\zv$ returned by \Cref{alg:batched} satisfies the following for each $t \in [T]$:
\begin{enumerate}
	\item For each $\ell \in [L]$, ${z}_{l,t} > 0 \implies \Phi_l(\zv_t) = \max_{l'} \Phi_{l'}(\zv_t)$.
	\item $\sum_l {z}_{l,t} < 1 - \frac{B}{2T} \implies p_t = 0$.
	\item $\sum_l {z}_{l,t} > 0 \implies p_t + q = \max_l \Phi_l(\zv_t)$. 
\end{enumerate}
\end{proposition}

It is easy to argue the existence of a greedy semi-allocation $\zv_t$ satisfying the three properties above using Kakutani's fixed point theorem. However, in \Cref{alg:batched}, we are able to constructively find such an allocation by maximizing a concave objective function subject to a budget-like constraint. To show that an optimal solution to this optimization problem satisfies the three properties in \Cref{prop:properties_batched}, we need the following technical lemma. 

\begin{proposition}
\label{prop:constructive}
Let $f_1, \ldots, f_n: \R^n \to \R$, such that $f_i = \frac{\partial F}{\partial x_i}$ for some concave differentiable function $F: \R^n \to \R$. Then, the solution $\xv^*$ to the optimization problem $\max\{ F(\xv): \xv \in \Delta^n\}$ satisfies, for all $i \in [n]$,
$$x_i^* > 0\implies f_i(\xv^*) = \max_{j \in [n]} f_j(\xv^*).$$
\end{proposition}
\begin{proof}
Written out in full form, the optimization problem is
\begin{align*}
    \max \quad  &F(\xv) \\
    \text{s.t.} \quad  & x_1 + \cdots + x_n = 1 \\
    &x_i \geq 0 \quad \forall \; i \in [n]
\end{align*}
Let $\lambda$ be the Lagrangian multiplier for the equality constraint and $(\theta_i: i \in [n])$ be the multipliers for the non-negativity constraints. The Lagrangian is given by $\mathcal{L}(\xv,\lambda,\thetav) = F(\xv) + \lambda (\sum_{i \in [n]} x_i - 1) + \sum_{i \in [n]} \theta_i x_i$. The KKT conditions require that there exist $\lambda, \thetav$ satisfying the following constraints:
\begin{itemize}
    \item (Primal feasibility) $\xv^* \in \Delta^n$,
    \item (Dual feasibility) $\theta_i^* \geq 0$ for all $i \in [n]$,
    \item (Complementary slackness) $x_i^*\theta_i^* = 0$ for all $i \in [n]$, 
    \item (Stationarity) $\frac{\partial F (\xv^*)}{\partial x_i} - \lambda^* + \theta^*_i = 0$ for all $i \in [n]$. 
\end{itemize}
Since $\frac{\partial F}{\partial x_i} = f_i$, the stationarity condition can be written as $f_i(\xv^*) = \lambda^* - \theta_i^*$. This implies $f_j(\xv^*) \leq \lambda^*$ for all $j$ since $\thetav^* \geq 0$. On the other hand, consider some $i$ with $x^*_i > 0$. By complementary slackness, $\theta^*_i = 0$, and hence $f_i(\xv^*) = \lambda^*$. Thus
$$x^*_i > 0 \implies f_i(\xv^*) = \lambda^* = \max_j f_j(\xv^*),$$
which completes the proof.
\end{proof}

We are now ready to prove \Cref{prop:properties_batched}.

\begin{proof}[Proof of \Cref{prop:properties_batched}]
Let $F(\zv, \lambda)$ be the objective of the optimization problem for computing the greedy semi-allocation in \Cref{alg:batched}:
\[
F(\zv,\lambda) = \frac{1}{N} \sum_{i=1}^N \ln\left(\tilde{u}_{i,t}(\zv) \right) + \lambda q.
\]

Define $f_1, \ldots, f_{L+1}: \R^{L+1} \to \R$ as follows:
\begin{itemize}
    \item $f_i( \zv, \lambda) = \Phi_i(\zv)$ for $i = 1, \ldots, L$, 
    \item $f_{L+1}(\zv, \lambda) = q$.
\end{itemize}

It can be checked that $F$ is concave and differentiable, $\frac{\partial F}{\partial z_i} = f_i$ for all $i \in [L]$, and that $\frac{\partial F}{\partial \lambda} = f_{L+1}$. Thus, applying \Cref{prop:constructive},\footnote{From the proof of \Cref{prop:constructive}, it can be seen that it still holds even if the constraint set is a scaling of the unit simplex.} the solution $(\zv_t, \lambda_t)$ to the optimization problem satisfies 
\begin{itemize}
    \item $\lambda_t > 0$ implies $q \geq \max_j\{\Phi_j(\zv_t)\}$, and
    \item For all $l \in [L]$, $z_{l,t} > 0$ implies $\Phi_l(\zv_t) = \max\{q, \max_j \Phi_j(\zv_t)\}$. 
\end{itemize}

We now check that $\zv_t$ satisfies the conditions in \Cref{prop:constructive}.
\begin{enumerate}
    \item $\sum_{l=1}^L z_{l,t} \leq 1 - \frac{B}{2T}$: This follows directly from the constraints of the optimization problem.
    \item $z_{l,t} > 0 \implies \Phi_l(\zv_t) = \max_j \{\Phi_j(\zv_t)\}$: This holds because $z_{l,t} > 0$ implies $\Phi_l(\zv_t) = \max\{q, \max_j \Phi_j(\zv_t)\}$. But since $\Phi_l(\zv_t)$ is a term in  $\max_j \Phi_j(\zv_t)$, this means $\Phi_l(\zv_t) = \max_j \Phi_j(\zv_t)$. 
    \item $\sum_l z_{l,t} < 1 - \frac{B}{2T} \implies \max_j \Phi_j(\zv_t) \leq q$: If $\sum_l z_{l,t} < 1 - \frac{B}{2T}$ then $\lambda_t > 0$, which implies $q \geq \max_j \Phi_j(\zv_t)$.
    \item $\sum_l z_{l,t} > 0 \implies q \leq \max_j \Phi_j(\zv_t)$: If $\sum_l z_{l,t} > 0$ then $z_{l,t} > 0 $ for some $l \in [L]$, which implies $\Phi_l(\zv_t) = \max\{q, \max_j \Phi_j(\zv_t)\}$. Thus $q \leq \max_j \Phi_j(\zv_t)$.
\end{enumerate}
This completes the proof.
\end{proof}

We can now prove our main performance guarantee for~\cref{alg:batched}, which we restate below.
\batched*
\begin{proof}
To prove the theorem, it suffices to check the following 3 statements:
\begin{enumerate}
	\item ~\cref{alg:batched} returns a feasible allocation $\xv$.
	\item The dual certificate $(p_t, q)$ in \cref{prop:properties_batched} is feasible to the dual LP $(D)$.
	\item 
$\sum_t p_t + Bq \leq \alpha$.
\end{enumerate}
Given these properties, the guarantee follows via weak LP duality for programs $(P)$ and $(D)$.

For Claim $(2)$, note that by~\Cref{lem:predicted_utilities}, predicted utilities are a lower bound on $c_iu_i(\xv)$ for all $i\in[N],t\in[T]$. Now by definition of our dual certificate, we have for all $i\in[N],t\in[T]$
\begin{align*}
p_t + q \geq \max_l \sum_{i=1}^N \frac{v_{i,l,t}}{\tilde{u}_{i,t}({\zv}_t)} \geq \max_l \sum_{i=1}^N \frac{v_{i,l,t}}{c_i {u}_{i,t}(\xv)}
\end{align*}
Claims $(1)$ and $(3)$ are implied by the following key invariant:
\begin{tcolorbox}
\textbf{Key Invariant}.\qquad $\sum_t (p_t + q) \sum_{l} {z}_{l,t} \leq \frac{\alpha}{4}$. 
\end{tcolorbox}
We first show how this key invariant implies Claims $(1)$ and $(3)$.
\begin{enumerate}
\item[(1)] Clearly $\xv$ satisfies the per-round constraints by definition of the algorithm, so we just need to show that it satisfies the overall budget constraint. Since the set-aside semi-allocation is chosen such that $\sum_{l,t} {y}_{l,t} = \frac{B}{2}$, we only need to check that $\sum_{l,t} {z}_{l,t} \leq \frac{B}{2}$. Now, since $p_t\geq 0\,\forall\,t$, the key invariant implies 
$$\sum_t\sum_l {z}_{l,t} \leq \frac{1}{q} \cdot \frac{\alpha}{4} = \frac{B}{2},$$
by the definition of $q=\frac{\alpha}{2B}$, and our choice of $\alpha$.
\item[(3)] By (the contrapositive of) item 2 of \Cref{prop:properties_batched}, we have that if $p_t > 0$ then $\sum_{l} {z}_{l,t} = 1 - \frac{B}{2T}$. Using this and the key invariant, we have
$$
\sum_t p_t \left(1 - \frac{B}{2T}\right) = \sum_tp_t \sum_l {z}_{l,t} \leq \frac{\alpha}{4}. 
$$
Moreover, since $B \leq T$, this implies $\sum_t p_t \leq \frac{\alpha}{2}$. Now, by our choice of $q=\frac{\alpha}{2B}$, we have
$$
\sum_t p_t + Bq \leq  \frac{\alpha}{2} + \frac{B\alpha}{2B}=\alpha.
$$
\end{enumerate}
Finally we turn to the
\textbf{proof of the key invariant.} We have
\begin{align*}
\sum_t (p_t + q) \sum_l {z}_{l,t}
 &=\sum_t \left(\frac{1}{N} \max_{l'} \sum_{i} \frac{v_{i,l',t}}{\tilde{u}_{i,t}({\zv}_t)}\right) \sum_l {z}_{l,t} \qquad\qquad\text{(by item 3 of \Cref{prop:properties_batched})} \\
 &= \frac{1}{N} \sum_t\sum_l {z}_{l,t} \cdot  \max_{l'} \sum_{i}\frac{v_{i,l',t}}{\tilde{u}_{i,t}({\zv}_t)} \\
 &= \frac{1}{N} \sum_t\sum_l {z}_{l,t} \cdot   \sum_{i} \frac{v_{i,l,t}}{\tilde{u}_{i,t}({\zv}_t)} 
 \qquad\qquad\qquad\;\text{(by item 1 of \Cref{prop:properties_batched})} \\
 &= \frac{1}{N} \sum_t \sum_i \frac{\sum_l {z}_{l,t}v_{i,l,t}}{\tilde{u}_{i,t}({\zv}_t)}
= \frac{1}{N} \sum_t \sum_i \frac{\tilde{u}_{i,t}({\zv}_t) - \tilde{u}_{i,t}(0)}{\tilde{u}_{i,t}({\zv}_t)} 
\end{align*}
Where the last equality follows from our choice of the predicted utility (\cref{def:predictedbatch}). Now, using the fact that $1 - x \leq - \ln(x)\,\forall\,x\in\mathbb{R}$, we have
\begin{align*}
\sum_t (p_t + q) \sum_l {z}_{l,t}
 &\leq \frac{1}{N} \sum_t \sum_i\left[ \ln\left(\tilde{u}_{i,t}({\zv}_t)\right) - \ln\left(\tilde{u}_{i,t}(0)\right) \right]\\
 &= \frac{1}{N} \sum_i\left[ \ln\left(\tilde{u}_{i,t}({\zv}_T)\right) - \ln\left(\tilde{u}_{i1}(0)\right) \right]  \qquad\qquad\qquad\;\text{(by telescoping)} \\
 &= \frac{1}{N} \sum_i \left[ 
 \ln\left(\frac{B}{2\min\{N, L\}T}\tilde{V}_i + \sum_t\sum_l v_{i,l,t}{z}_{l,t}
 \right) 
 -\ln\left( \frac{B}{2\min\{N, L\}T} \tilde{V}_i \right)
 \right] \\
 &= \ln\left(\frac{2\min\{N, L\}T}{B} \right)
 + \frac{1}{N} \sum_i 
 \ln\left( 
 \frac{B}{2\min\{N, L\}T} + \frac{1}{\tilde{V}_i}\sum_t\sum_l v_{i,l,t}{z}_{l,t}
 \right) \\
 &\leq \ln\left(\frac{2\min\{N, L\}T}{B} \right)
 + \frac{1}{N} \sum_i 
 \ln\left( 
 \frac{B}{2T} + \frac{1}{\tilde{V}_i}\sum_t\sum_l v_{i,l,t}{z}_{l,t}
 \right)
\end{align*}
Now, observe that
\begin{align*}
\sum_t\sum_l v_{i,l,t}{z}_{l,t}
&\leq \sum_t\max_{l'} \{v_{i,l',t}\}\sum_l{z}_{l,t} \\
&\leq \sum_t\max_{l'} \{v_{i,l',t}\}  \left(1 - \frac{B}{2T}\right)
= V_i\left(1 - \frac{B}{2T}\right),
\end{align*}
where the second inequality is because $\sum_l {{z}_{l,t}} \leq 1 - \frac{B}{2T}$, by the stopping condition. 
Thus
\begin{align*}
    \frac{B}{2T} + \frac{1}{\tilde{V}_i}\sum_t\sum_l v_{i,l,t}{z}_{l,t}
    &\leq \frac{B}{2T} + \left(1 - \frac{B}{2T}\right)\frac{V_i}{\tilde{V}_i} 
    \leq \frac{B}{2T} + \left(1 - \frac{B}{2T}\right)d_i 
    \leq d_i.
\end{align*}
Therefore, we conclude that
$$\sum_{t} (p_t+q) \sum_l {z}_{l,t} 
\leq 
\ln\left(\frac{2\min\{N, L\}T}{B} \right) 
+ \frac{1}{N} \sum_i \ln(d_i) \leq \frac{\alpha}{4},$$
which concludes the proof of the key claim.
\end{proof}

\section{Discussion and Future Directions}\label{sec:discussion}

\textbf{The Geometry of Online Fair Allocation.} 
While we focus on the proportional fairness objective, a natural question is whether our results can be extended to other non-linear objective functions, such as the class of generalized $p$-mean welfare measures~\cite{BKM21,ebadian2022efficient}. Similarly, while our results give insight into the interaction between per-round budget constraints and overall budget constraints, one can ask if these ideas can be extended to more complex constraints such as general packing constraints~\cite{FMS18}. 

\textbf{Alternate Information Structures.} 
An interesting direction for future research is to develop a better understanding regarding whether what types of predictions would be the most appropriate for online fair division. For example, what if the algorithm is also provided with a prediction regarding the total value of each good across all agents, but not specifying which of the agents will like it and by how much? Could this additional information allow us to overcome the logarithmic lower bounds and maybe even achieve a constant approximation? Would it help to know more detailed patterns about agent valuations --- for example, if they take discrete values, or are periodic, or have low variance? 

\textbf{Alternate Arrival Models.} 
A related question to the one above is in regards to the process for generating agent valuations. In our model, the values $V = (v_{i,t})_{i \in [N], t \in [T]}$ are allowed to arrive in an adversarial order. What if we consider slightly more optimistic models such as the random-order model or the stochastic model~\cite{borodin2005online}? It is possible that predictions are not needed when studying these arrival models since they inherently already give the algorithm some information about the input.


\textbf{Incentives.} Finally, one of the most important open questions is how to convert online algorithms to online mechanisms, whereby agents report their private $v_{i,t}$ values in each round. In such a setting, can we design mechanisms which are incentive compatible? There has been recent work on similar questions for static fair allocation~\cite{2017nashwelfare,amanatidis2017truthful,halpern2020fair,ALW19}, as well as 
online welfare maximization~\cite{guo2009,balseiro2019dynamic,gorokh2017,gorokh2021remarkable}; extending these ideas to online fair allocation is an important open problem.






\bibliographystyle{unsrtnat}
\bibliography{abb,main-bibliography,MARA}



\newpage
\appendix

\section*{Appendix}






\section{Proof of Theorem~\ref{thm:binarypf}}
\begin{proof}
Let $\bx$ denote~\cref{alg:Binary}'s final allocation, and consider the following LP and its dual:
\begin{equation*}
\begin{aligned}[t]
(P)\qquad 
&\max_{\wv\in\mathbb{R}_+^{T}}\ \frac{1}{N}\sum_{i=1}^N \frac{u_i( \wv)}{u_i( \xv)} \\
&\text{s.t.}\ \sum_{t=1}^T w_{t} \leq 1 
\end{aligned}
\qquad \qquad
\begin{aligned}[t]
\label{eq:price_lp}
(D) \qquad&\min_{p\in \mathbb{R}_+}\ p \\
&\text{s.t.} \ p \geq \frac{1}{N}\sum_{i=1}^N \frac{v_{i,t}}{u_i( \bx)}\quad\; \forall\,t\in[T]\\
\end{aligned}
\end{equation*}
By construction, $\gamma$ is feasible to the dual LP, because
$$\gamma 
\stackrel{(a)} \geq \frac{1}{N} \sum_{i=1}^N \frac{v_{it}}{\tilde{u}_{it}(z_t)}
 \stackrel{(b)}\geq \max_i \frac{1}{N}\sum_{i=1}^N \frac{v_{it}}{u_i({\xv})},
$$
where $(a)$ is by definition of the algorithm, and $(b)$ is because promised utilities are a lower bound on the true final utility. Thus by weak duality,
$$\max_{\wv\in\mathbb{R}_+^{T}}\ \frac{1}{N}\sum_{i=1}^N \frac{u_i( \wv)}{u_i(\xv)} \leq \gamma = 2\ln(2N).$$
To complete the proof, it remains to show that $\xv$ is a feasible allocation.
Since $\sum_{t=1}^T y_t \leq \frac12$, it suffices to show that $\sum_{t=1}^T z_t \leq \frac12$. Now we have 
\begin{align*}
    \gamma \sum_{t=1}^T {z}_t
    &= \sum_{t=1}^T {z}_t \cdot \frac{1}{N} \sum_{i=1}^N \frac{v_{it}}{\tilde{u}_{it}{({z}_t)}} \hspace{2.7cm}\big(\text{since ${z}_t > 0$ implies $\gamma =\frac{1}{N} \sum_{i=1}^N \frac{v_{it}}{\tilde{u}_{it}(z_t)} $}\big) \\
    &= \frac{1}{N} \sum_{i=1}^N \sum_{t=1}^T \frac{{z}_t v_{it}}{\tilde{u}_{it}({z}_t)}
    = \frac{1}{N} \sum_{i=1}^N \sum_{t=1}^T \left(1 - \frac{\tilde{u}_{it}(0)}{\tilde{u}_{it}({z}_t)} \right) \\
    &\leq \frac{1}{N} \sum_{i=1}^N \sum_{t=1}^T \left[\ln(\tilde{u}_{it}({z}_t)) - \ln(\tilde{u}_{it}(0))\right]
    \qquad\text{(since $1 - x \leq - \ln(x)\,\forall\,x\in\mathbb{R}$)}\\
    &= \frac{1}{N} \sum_{i=1}^N \left[\ln(\tilde{u}_{iT}({z}_T)) -\ln(\tilde{u}_{i1}(0))  \right]
    = \frac{1}{N} \sum_{i=1}^N \left[\ln(\tilde{u}_{iT}({z}_T)) -\ln\left(\frac{1}{2N}\right) \right] \\
    &= \frac{1}{N} \sum_{i=1}^N  \ln\left(\tilde{u}_{i}({\xv})\right)+\ln(2N)  \qquad\qquad\text{(since predicted utility lower bounds final utility)}\\
    &\leq \ln(2N). \hspace{4cm}\text{(since final utilities are at most 1)}
 \end{align*}
 Hence, in order for $\sum_{t=1}^T {z}_t \leq \frac12$, it suffices to have $\gamma = 2\ln(2N)$. This is why we defined $\gamma$ to be equal to this quantity in the algorithm.  
\end{proof}




\section{Proof of Theorem~\ref{thm:general}}
As indicated in \Cref{sec:general_alg}, we only need to prove that the greedy semi-allocations combinatorially computed in \Cref{alg:General-L-1} still satisfy the properties in \Cref{prop:properties_batched}. The result below states these properties for $L=1$. Note that the first property in \Cref{prop:properties_batched} becomes vacuous for $L=1$.

\begin{proposition}[Some Properties of the Algorithm]
\label{prop:properties}
Consider the dual certificate:
\begin{align*}
q &= \frac{\alpha}{2B} \geq \frac{2}{B} \left[\ln\left(\frac{2T}{B}\right) + \frac{1}{N} \sum_i \ln(d_i)\right]\\
p_t &= \max\left\{0, \frac{1}{N} \sum_{i=1}^N \frac{v_{i,t}}{\tilde{u}_{i,t}(\zv_t)} - q\right\}.
\end{align*}
Then, the allocation $\xv=\yv+\zv$ returned by \Cref{alg:General-L-1} satisfies the following:
\begin{enumerate}
	\item ${z}_{t} < 1 - \frac{B}{2T} \implies p_t = 0$.
	\item ${z}_{t} > 0 \implies p_t + q =  \sum_{i=1}^N \frac{v_{i,t}}{\tilde{u}_{i,t}({z}_t)}$. 
\end{enumerate}
\end{proposition}
\begin{proof}
(1) follows because if $z_t < 1 -\frac{B}{2T}$, then that must mean $z_t = z^*_t$. Hence, $\frac{1}{N} \sum_{i=1}^N \frac{v_{i,t}}{\tilde{u}_{i,t}(z_t)} \leq \frac{\alpha}{2B} = q$, which implies $p_t = 0$. 

(2) follows because if $z_t > 0$, then that must mean $z^*_t > 0$. Let $\Phi(z) = \frac{1}{N} \sum_{i=1}^N \frac{v_{i,t}}{\tilde{u}_{i,t}(z)}$. Since $\Phi$ is a continuous function in $z$, this implies  $\Phi(z^*_t) = \frac{\alpha}{2B} = q$. Because $z_t  = \min\{z^*_t, 1 - y_t\} \leq z^*_t$ and $\Phi$ is decreasing, we have $\Phi(z_t) \geq \Phi(z^*_t) = q$. Thus $p_t = \max\{0, \Phi(z_t) - q\} = \Phi(z_t) - q$.  
\end{proof}




\end{document}